\documentclass[a4paper, USenglish, cleveref, autoref, thm-restate, numberwithinsect,
]{lipics-v2021}

\pdfoutput=1 
\hideLIPIcs  

\title{Improved Algorithm for Reachability in $d$-VASS}

\titlerunning{Improved Algorithm for Reachability in $d$-VASS} 

\author{Yuxi Fu}{BASICS, Shanghai Jiao Tong University}{fu-yx@cs.sjtu.edu.cn}{}{}
\author{Qizhe Yang}{Shanghai Normal University}{qzyang@shnu.edu.cn}{https://orcid.org/0009-0000-9010-5364}{}
\author{Yangluo Zheng}{BASICS, Shanghai Jiao Tong University}{wunschunreif@sjtu.edu.cn}{https://orcid.org/0009-0000-1028-5458}{}

\authorrunning{Y. Fu and Q. Yang and Y. Zheng} 

\Copyright{Yuxi Fu, Qizhe Yang, Yangluo Zheng} 

\ccsdesc[500]{Theory of computation~Models of computation}
\ccsdesc[500]{Theory of computation~Logic and verification}

\keywords{Petri net, vector addition system, reachability} 

\relatedversion{} 

\acknowledgements{We thank Weijun Chen, Huan Long, Hao Wu and Qiang Yin for proofreading various versions of this paper.}

\nolinenumbers 

\EventEditors{Karl Bringmann, Martin Grohe, Gabriele Puppis, and Ola Svensson}
\EventNoEds{4}
\EventLongTitle{51st International Colloquium on Automata, Languages, and Programming (ICALP 2024)}
\EventShortTitle{ICALP 2024}
\EventAcronym{ICALP}
\EventYear{2024}
\EventDate{July 8--12, 2024}
\EventLocation{Tallinn, Estonia}
\EventLogo{}
\SeriesVolume{297}
\ArticleNo{127}

\usepackage{bm}
\usepackage{stmaryrd}
\usepackage{tikz}
\usetikzlibrary{shapes}

\begin{document}

\maketitle

\begin{abstract}
    An $\mathsf{F}_{d}$ upper bound for the reachability problem in vector addition systems with states (VASS) in fixed dimension is given, where $\mathsf{F}_d$ is the $d$-th level of the Grzegorczyk hierarchy of complexity classes.
The new algorithm combines the idea of the linear path scheme characterization of the reachability in the $2$-dimension VASSes with the general decomposition algorithm by Mayr, Kosaraju and Lambert.
The result improves the $\mathsf{F}_{d + 4}$ upper bound due to Leroux and Schmitz (LICS 2019). 
\end{abstract}

\newcommand{\defproblem}[3]{
    \noindent\fbox{
        \begin{minipage}{0.96\textwidth}
            \underline{#1}
            \begin{description}
                \item[Input:] {#2}
                \item[Question:] {#3}
            \end{description}
        \end{minipage}
    }
}

\renewcommand{\Vec}[1]{\bm{#1}}
\newcommand{\norm}[1]{\left\Vert{#1}\right\Vert}
\newcommand{\Supp}{\mathop{\operatorname{supp}}\nolimits}
\newcommand{\CycleSpace}{\operatorname{Cyc}\nolimits}
\newcommand{\Span}{\operatorname{span}\nolimits}
\newcommand{\ActWord}[1]{\llbracket{#1}\rrbracket}
\newcommand{\Cone}{\operatorname{cone}\nolimits}
\newcommand{\Next}{\operatorname{next}\nolimits}
\newcommand{\Rank}{\operatorname{rank}\nolimits}
\newcommand{\RankFull}{\operatorname{rank}\nolimits_{\text{full}}}
\newcommand{\ExpF}[1]{{\normalfont\textsf{exp}({#1})}}
\newcommand{\PolyF}[1]{{\normalfont\textsf{poly}({#1})}}
\newcommand{\Clean}[1]{{\normalfont\textrm{clean}({#1})}}
\newcommand{\Dec}[1]{{\normalfont\textrm{dec}({#1})}}
\newcommand{\Facc}{\textsc{Facc}}
\newcommand{\Bacc}{\textsc{Bacc}}

\newcommand{\bracket}[1]{\left\langle{#1}\right\rangle}
\newcommand{\lelex}{\le_{\text{lex}}}
\newcommand{\ltlex}{<_{\text{lex}}}

\newcommand{\LPSSystem}[1]{\mathcal{E}_{\normalfont\text{LPS}}({#1})}
\newcommand{\LPSSystemHomo}[1]{\mathcal{E}^0_{\normalfont\text{LPS}}({#1})}

\newcommand{\KLMSystem}[1]{\mathcal{E}({#1})}
\newcommand{\KLMSystemHomo}[1]{\mathcal{E}^0({#1})}

\section{Introduction}

Petri nets, or equivalently vector addition system with states (VASS), are a well studied model of concurrency. A VASS consists of a finite state control where each state transition has as its effect an integer-valued vector, and its configurations are pairs of a state and a vector with \emph{natural number} components. A transition may lead from one configuration to another by adding its effect component-wise, conditioned on that the components of the resulting vector remain non-negative. The \emph{reachability problem}, which asks whether from one configuration there is a path reaching another configuration, lies in the center of the algorithmic theory of Petri nets and has found a wide range of applications due to its generic nature. Since the problem was shown to be decidable by Mayr \cite{Mayr81} in 1981, its computational complexity had been a long-standing open problem in the field. In 2015, Leroux and Schmitz \cite{LS15} presented the first complexity upper bound, stating that the reachability problem is cubic-Ackermannian. This was later improved to an Ackermannian upper bound in 2019, again by Leroux and Schmitz \cite{LS19}. Regarding the hardness, in 2021 seminal works by Czerwi\'nski and Orlikowski \cite{CO21}, and independently by Leroux \cite{Leroux21}, provided matching Ackermannian lower bounds, settling the exact complexity of the problem.

Concerning the parameterization by dimension, i.e.\ the reachability problem in $d$-dimensional VASSes where $d$ is fixed, there is still a gap in the known complexity bounds. Currently, we only have exact complexity bounds for dimension one and two \cite{HKOW09,BEF+21}. For dimension $d \ge 3$, the result of Leroux and Schmitz \cite{LS19} shows that the problem is in $\mathsf{F}_{d + 4}$, the $(d{+}4)$-th level of the Grzegorczyk hierarchy of complexity classes. Note that a recent work by Yang and Fu \cite{YF23} points out that the problem for dimension $3$ is in $\mathsf{F}_3 = \mathsf{TOWER}$. On the other hand, the best known lower bound by Czerwi\'nski, Jecker, Lasota, Leroux and Orlikowski \cite{CJLLO23} states that reachability in $(2d{+} 3)$-dimensional VASSes is $\mathsf{F}_d$-hard. Motivated by this gap, our paper focuses on the computational complexity of reachability problem in the fixed-dimensional VASSes.

\paragraph*{Our contribution.}

In this paper we show that the reachability problem in the $d$-dimensional VASS is in $\mathsf{F}_d$ for $d\ge 3$, improving the previous $\mathsf{F}_{d + 4}$ upper bound by Leroux and Schmitz \cite{LS19}, and generalizing the tower upper bound for the reachability problem in $3$-VASS~\cite{YF23}.
The new upper bound is obtained with the help of two novel technical lemmas.
\begin{enumerate}
\item 
Our main technical tool (Theorem~\ref{thm:geo-2vass-flat-lps}) is a generalization of the linear path scheme characterization for the reachability relation in the $2$-dimensional VASSes~\cite{BEF+21}.
By borrowing the key idea from the work of Yang and Fu~\cite{YF23}, we show that as long as the ``geometric dimension'' of a VASS (that is, the dimension of the vector space spanned by the effects of cyclic paths) is bounded by $2$, its reachability relation can be characterized by short linear path schemes.
We then apply the lemma to simplify the KLMST algorithm so that (i) a VASS is replaced by a short linear path scheme whenever its geometric dimension is no more than $2$ and (ii) the linear path schemes will not be decomposed further.
It is then routine~\cite{LS19}, using the tools from~\cite{Schmitz16}, to show that the reachability problem in the $d$-dimensional VASS is in $\mathsf{F}_{d+1}$ for all $d\ge 3$.
\item
Our second lemma (Lemma~\ref{lem:rel-fast-grow-func}) allows us to improve further the bound from $\mathsf{F}_{d+1}$ to $\mathsf{F}_{d}$.
This is done by a careful analysis of the properties of the fast-growing functions~\cite{Schmitz16}.
\end{enumerate}
Due to space limitation the proofs of the two lemmas are placed in the appendices. 

\paragraph*{Organization.}
Section \ref{sec:prelim} fixes notation, defines the VASS model and its reachability problem.
Section \ref{sec:flat-geo-2d} generalizes the linear path scheme characterization \cite{BEF+21} to VASSes whose geometric dimension are bounded by $2$. Section \ref{sec:lps-system} recalls the characterization system of linear inequalities for linear path schemes.
Section \ref{sec:mod-klmst} makes use of the results of Section \ref{sec:flat-geo-2d} to give an improved version of the classic KLMST decomposition algorithm.
Section \ref{sec:complexity} analyzes the complexity of our modified algorithm, proving the main result. Section \ref{sec:conclusion} concludes. Proofs omitted from the main text can be found in the appendices.

\section{Preliminaries}
\label{sec:prelim}

We use $\mathbb{N}, \mathbb{Z}, \mathbb{Q}$ to denote respectively the set of non-negative integers, integers, and rational numbers.
Let $n \in \mathbb{N}$ be a number, we write $[n]$ for the range $\{1, 2, \ldots, n\}$.
Let $\Vec{u}, \Vec{v} \in \mathbb{X}^d$ be $d$-dimensional vectors where $\mathbb{X}$ can be any one of $\mathbb{N}, \mathbb{Z}, \mathbb{Q}$.
We write $\Vec{v}(i)$ for the $i$-th component of $\Vec{v}$ where $i \in [d]$, so $\Vec{v} = (\Vec{v}(1), \ldots, \Vec{v}(d))$.
The \emph{maximum norm} of $\Vec{v}$ is defined to be $\norm{\Vec{v}} := \max_{i \in [d]}|\Vec{v}(i)|$.
We extend component-wise the order $\le$ for vectors in $\mathbb{X}^d$, so $\Vec{u} \le \Vec{v}$ if and only if $\Vec{u}(i) \le \Vec{v}(i)$ for all $i \in [d]$.
Addition and subtraction of vectors are also component-wise, so $(\Vec{u} + \Vec{v})(i) = \Vec{u}(i) + \Vec{v}(i)$ for all $i \in [d]$.
Define $\Supp(\Vec{v}) := \{i \in [d] : \Vec{v}(i) \ne 0\}$ to be the set of indices of non-zero components of $\Vec{v}$. This notation is extended to sets of vectors naturally, so $\Supp(S) = \bigcup_{\Vec{v} \in S}\Supp(\Vec{v})$ for any set $S\subseteq \mathbb{X}^d$.

For technical reasons we introduce the symbol $\omega$ that stands for the infinite element. Let $\mathbb{N}_\omega := \mathbb{N} \cup \{\omega\}$.
We stipulate that $n < \omega$ for all $n\in \mathbb{N}$, and $x + \omega = \omega + x = \omega$ for all $x \in \mathbb{Z}$.
Define the partial order $\sqsubseteq$ over $\mathbb{N}_\omega$ so that $x \sqsubseteq y$ if and only if $x = y$ or $y = \omega$ for all $x, y \in \mathbb{N}_\omega$.
The relation $\sqsubseteq$ is also extended component-wise to vectors in $\mathbb{N}_\omega^d$.

Let $\Sigma$ be a finite alphabet and $s, t \in \Sigma^*$ be two strings over $\Sigma$. We write $st$ for the concatenation of $s$ and $t$, and $s^n$ for the concatenation of $n$ copies of $s$ where $n \in \mathbb{N}$. If $s = a_1a_2\ldots a_\ell$ where $a_1, \ldots, a_\ell \in \Sigma$, we write $|s| := \ell$ for the length of $s$, and $s[i\ldots j] := a_ia_{i+1}\ldots a_j$ for the substring of $s$ between indices $i$ and $j$ where $1 \le i \le j \le |s|$.

\subsection{Vector Addition Systems with States}

Let $d \ge 0$ be an integer. A \emph{$d$-dimensional vector addition system with states ($d$-VASS)} is a pair $G = (Q, T)$ where $Q$ is a finite set of \emph{states} and $T\subseteq Q\times \mathbb{Z}^d \times Q$ is a finite set of \emph{transitions}.
Clearly a VASS can also be viewed as a directed graph with edges labelled by integer vectors.
Given a word $\pi = (p_1, \Vec{a}_1, q_1) (p_2, \Vec{a}_2, q_2) \ldots (p_n, \Vec{a}_n, q_n) \in T^*$ over transitions, we say that $\pi$ is a \emph{path from $p_1$ to $q_n$} if $q_{i} = p_{i + 1}$ for all $i = 1, \ldots, n - 1$. It is a \emph{cycle} if we further have $p_1 = q_n$.
The \emph{effect} of $\pi$ is defined to be $\Delta(\pi) := \sum_{i = 1}^n \Vec{a}_i$, and the \emph{action word} of $\pi$ is the word $\ActWord{\pi} := \Vec{a}_1\Vec{a}_2\ldots\Vec{a}_n$ over $\mathbb{Z}^d$.
The \emph{Parikh image} of $\pi$ is the function $\phi \in \mathbb{N}^T$ mapping each transition to its number of occurrences in $\pi$. Given a function $\phi \in \mathbb{N}^T$ we also define $\Delta(\phi) := \sum_{t = (p, \Vec{a}, q)\in T} \phi(t)\cdot \Vec{a}$. Note that $\Delta(\phi) = \Delta(\pi)$ if $\phi$ is the Parikh image of $\pi$.
Let $L \subseteq T^*$ be a language (i.e.\ subset of words), we define its effect as $\Delta(L) := \{\Delta(\pi): \pi \in L\}$.

The norm of a transition $t = (p, \Vec{a}, q)$ is defined by $\norm{t} := \norm{\Vec{a}}$. The norm of a path $\pi = t_1t_2\ldots t_n$ is $\norm{\pi} := \max_{i \in [n]} \norm{t_i}$.
For a $d$-VASS $G = (Q, T)$ we write $\norm{T} := \max\{\norm{t} :t \in T\}$. The \emph{size} of $G$ is defined by
\begin{equation}
    \label{eq:vass-size}
    |G| := |Q| + |T| + d\cdot |T| \cdot \norm{T}.
\end{equation}

\paragraph*{Semantics of VASSes.}

Let $G = (Q, T)$ be a $d$-VASS.
A \emph{configuration} of $G$ is a pair of a state $p \in Q$ and a vector $\Vec{v}\in \mathbb{Z}^d$, written as $p(\Vec{v})$. Let $\mathbb{D}\subseteq \mathbb{Z}^d$, we define the $\mathbb{D}$-semantics for $G$ as follows.
For each transition $t = (p, \Vec{a}, q) \in T$, the one-step transition relation $\xrightarrow{t}_{\mathbb{D}}$ relates all pairs of configurations of the form $(p(\Vec{u}),q(\Vec{v}))$ where $\Vec{u}, \Vec{v}\in\mathbb{D}$ and $\Vec{v} = \Vec{u} + \Vec{a}$.
Then for a word $\pi = t_1 t_2 \ldots t_n \in T^*$, the relation $\xrightarrow{\pi}_{\mathbb{D}}$ is the composition $\xrightarrow{\pi}_{\mathbb{D}} {}:={} \xrightarrow{t_1}_{\mathbb{D}}\circ \cdots \circ \xrightarrow{t_n}_{\mathbb{D}}$.
So $p(\Vec{u}) \xrightarrow{\pi}_{\mathbb{D}} q(\Vec{v})$ if and only if there are configurations $p_0(\Vec{u}_0), \ldots, p_n(\Vec{u}_n) \in Q\times \mathbb{D}$ such that
\begin{equation}
    p(\Vec{u}) = p_0(\Vec{u}_0) \xrightarrow{t_1}_{\mathbb{D}} p_1(\Vec{u}_1) \xrightarrow{t_2}_{\mathbb{D}} \cdots \xrightarrow{t_n}_{\mathbb{D}} p_n(\Vec{u}_n) = q(\Vec{v}).
\end{equation}
Also, when $\pi = \epsilon$ is the empty word, the relation $\xrightarrow{\epsilon}_{\mathbb{D}}$ is the identity relation over $Q\times \mathbb{D}$.
Note that $\xrightarrow{\pi}_{\mathbb{D}}$ is non-empty only if $\pi$ is a path.
When $p(\Vec{u}) \xrightarrow{\pi}_{\mathbb{D}} q(\Vec{v})$ we also say that $\pi$ induces (or is) a \emph{$\mathbb{D}$-run} from $p(\Vec{u})$ to $q(\Vec{v})$. We emphasize that all configurations on this run lie in $\mathbb{D}$, and that they are uniquely determined by $p(\Vec{u})$ and $\pi$.
For a language $L\subseteq T^*$ we define $\xrightarrow{L}_{\mathbb{D}}$ as $\bigcup_{\pi\in L}\xrightarrow{\pi}_{\mathbb{D}}$.
Finally, the \emph{$\mathbb{D}$-reachability relation} of $G$ is defined to be $\xrightarrow{*}_{\mathbb{D}} {}:={} \xrightarrow{T^*}_{\mathbb{D}}$.

For the above definitions, we shall often omit the subscript $\mathbb{D}$ when $\mathbb{D} = \mathbb{N}^d$.

We mention that in Section \ref{sec:mod-klmst} we need to generalize the VASS semantics to configurations in $Q\times \mathbb{N}_\omega^d$, allowing $\omega$ components in vectors. The definitions of $\xrightarrow{t}_{\mathbb{N}_\omega^d}$, $\xrightarrow{\pi}_{\mathbb{N}_\omega^d}$, and $\xrightarrow{*}_{\mathbb{N}_\omega^d}$ are similar to the above.

\paragraph*{Reachability problem.}

The reachability problem in vector addition systems with states is formulated as follows:

\defproblem{\textsc{Reachability in $d$-Dimensional Vector Addition System with States}}
    {A $d$-dimensional VASS $G = (Q, T)$, two configurations $p(\Vec{u}), q(\Vec{v}) \in Q\times \mathbb{N}^d$.}
    {Does $p(\Vec{u})\xrightarrow{*}_{\mathbb{N}^d} q(\Vec{v})$ hold in $G$?}

Note that we study the reachability problem for fixed-dimensional VASSes, where the dimension $d$ is treated as a constant to allow more fine-grained analysis. So we shall use the big-$O$ notation to hide constants that may depend on $d$. The general problem where the dimension can be part of the input was already shown to be Ackermann-complete \cite{CO21,Leroux21}.

\paragraph*{Cycle spaces and geometric dimensions.}

One of the key insights of \cite{LS19} is a new termination argument for the KLMST decomposition algorithm based on the dimensions of vector spaces spanned by cycles in VASSes, which yielded the primitive recursive upper bound of VASS reachability problem in fixed dimensions.
The vector spaces spanned by cycles also play an important role in our work.

\begin{definition}
    Let $G$ be a $d$-VASS. The \emph{cycle space} of $G$ is the vector space $\CycleSpace(G)\subseteq \mathbb{Q}^d$ spanned by the effects of all cycles in $G$, that is:
    \begin{equation}
        \CycleSpace(G) := \Span\{\Delta(\beta): \beta \text{ is a cycle in } G\}.
    \end{equation}
    The dimension of the cycle space of $G$ is called the \emph{geometric dimension} of $G$. We say $G$ is \emph{geometrically $k$-dimensional} if $\dim(\CycleSpace(G)) \le k$.
\end{definition}

\section{Flattability of Geometrically 2-dimensional VASSes}
\label{sec:flat-geo-2d}

A VASS is \emph{flat} if each of its states lies on at most one cycle. Flat VASSes form an important subclass of VASSes due to its connection to Presburger arithmetic, and we refer the readers to \cite{Leroux21b} for a survey.
In dimension $2$, it was proved in \cite{BEF+21} that $2$-VASSes enjoy a stronger form of flat representation, known as the \emph{linear path schemes}.
A linear path scheme is a regular expression of the form $\alpha_0 \beta_1^* \alpha_1 \ldots \beta_k^* \alpha_k$ where $\alpha_0, \ldots \alpha_k$ are paths and $\beta_1, \ldots, \beta_k$ are cycles of the VASS, such that they form a path when joined together. The results of \cite{BEF+21} show that the reachability relation of every $2$-VASS can be captured by short linear path schemes.

Linear path schemes are extremely useful as they can be fully characterized by linear inequality systems so that standard tools for linear or integer programming can be applied.
In this section, we generalize the results in \cite{BEF+21} and show that the reachability relation of any $d$-VASS whose geometric dimension is bounded by $2$ can also be captured by short linear path schemes.


Our proof follows closely to the lines of \cite{BEF+21}. Given a geometrically $2$-dimensional VASS $G$, we first cover $\mathbb{N}^d$ by the following two regions: one for the region far away from every axis:
\begin{equation}
    \mathbb{O} := \{ \Vec{u} \in \mathbb{N}^d : \Vec{u}(i) \ge D \text{ for all } i \in [d]\}
\end{equation}
where $D$ is some properly chosen threshold; the other one for the region close to some axis:
\begin{equation}
    \mathbb{L} := \{ \Vec{u} \in \mathbb{N}^d : \Vec{u}(i) \le D' \text{ for some } i \in [d] \}
\end{equation}
where $D'$ is chosen slightly above $D$ to create an overlap with $\mathbb{O}$.
Let $\pi$ be a run that we are going to capture by linear path schemes. We can extract its maximal prefix that lies completely in either $\mathbb{O}$ or $\mathbb{L}$, depending on where $\pi$ starts. This prefix must end at a configuration that lies in $\mathbb{L}\cap \mathbb{O}$, if we haven't touched the end of $\pi$. From this configuration we then extract a maximal cycle that also ends in $\mathbb{L}\cap \mathbb{O}$. Continuing this fashion, we can break $\pi$ into segments of runs that lie completely in $\mathbb{O}$ or $\mathbb{L}$, interleaved by cycles that start and end in $\mathbb{O}$ (actually in $\mathbb{L}\cap \mathbb{O}$). Note that the number of such cycles cannot exceed the number of states of $G$ since they are maximal.
Now we only need to capture the following three types of runs by short linear path schemes:
\begin{enumerate}
    \item Runs that are cycles starting and ending in $\mathbb{O}$.

        This will be handled in Section \ref{sec:flat-high-runs} in the appendix. Briefly speaking, since the geometric dimension of $G$ is $2$, the effect of such a cycle must belong to a plane in $\mathbb{Z}^d$. We will find a clever way to project this plane to a coordinate plane, and then project the $d$-VASS $G$ onto this plane to get a $2$-VASS.
        This is made possible by a novel technique called the ``sign-reflecting projection'' developed in Section \ref{sec:srp}. Intuitively speaking, for any vector in a plane we are able to determine whether it belongs to a certain orthant by observing only $2$ entries of this vector. The $d$-VASS can then be projected onto these $2$ coordinates. (See Example \ref{eg:sign-refl-proj} for a more concrete demonstration.)
        Now we apply the results of \cite{BEF+21} to obtain a linear path scheme that captures the projected cyclic run. The projection guarantees that we can safely project it back to get a linear path scheme for the run in $G$.
    \item Runs that lie completely in $\mathbb{O}$.

        This is just an easy corollary of the first type, and will also be handled in Section \ref{sec:flat-high-runs}. Just note that any run can be broken into a series of simple paths interleaved by cycles.
    \item Runs that lie completely in $\mathbb{L}$.

        This will be handled in Section \ref{sec:flat-low-runs}, by a long and tedious case analysis. In principle, we are going to argue that any such run essentially corresponds to a run in some $(d-1)$-VASS, so that we can use induction.
\end{enumerate}

\begin{example}
    \label{eg:sign-refl-proj}
    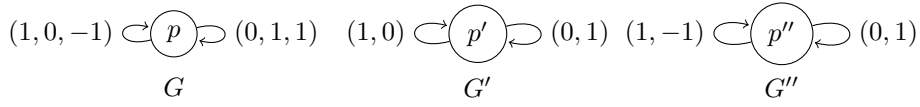
\begin{figure}[t]
        \centering
        \begin{tikzpicture}
            \path (0, 0) node[draw, circle] (p) {$p$};
            \path (p) edge [loop left] node[left]{$(1, 0, -1)$} (p);
            \path (p) edge [loop right] node[right]{$(0, 1, 1)$} (p);
            \path (0, -0.7) node {$G$};

            \path (4, 0) node[draw, circle] (p1) {$p'$};
            \path (p1) edge [loop left] node[left]{$(1, 0)$} (p1);
            \path (p1) edge [loop right] node[right]{$(0, 1)$} (p1);
            \path (4, -0.7) node {$G'$};

            \path (8, 0) node[draw, circle] (p2) {$p''$};
            \path (p2) edge [loop left] node[left]{$(1, -1)$} (p2);
            \path (p2) edge [loop right] node[right]{$(0, 1)$} (p2);
            \path (8, -0.7) node {$G''$};
        \end{tikzpicture}
        \caption{A geometrically $2$-dimensional VASS $G$ and two possible projections of it.}
        \label{fig:eg-sign-refl-proj}
    \end{figure}

    Consider the geometrically $2$-dimensional $3$-VASS $G$ as shown in Figure \ref{fig:eg-sign-refl-proj}. 
    It consists of a single state $p$ and two transitions with effects $(1, 0, -1)$ and $(0, 1, 1)$. 
    In order to apply the results of \cite{BEF+21}, one would like to derive a $2$-VASS from $G$ that reflects runs in $G$. A simple idea is to discard one coordinate of $G$. Two possibilities of this idea are shown in Figure \ref{fig:eg-sign-refl-proj} as $G'$, which discards the third coordinate, and $G''$, which discards the second coordinate. However, not all of them are satisfactory. For example, the legal run $p(0, 0) \xrightarrow{(1, 0)} p(1, 0)$ in $G'$ reflects an illegal run $p(0, 0, 0) \xrightarrow{(1, 0, -1)} p(1, 0, -1)$ in $G$ where the third coordinate goes negative. On the other hand, all runs in $G''$ can be safely projected back to a run in $G$. To see this, just observe that for any vector $\Vec{v}$ in the linear span of $(1, 0, -1)$ and $(0, 1, 1)$, $\Vec{v} \ge \Vec{0}$ if and only if $\Vec{v}(1) \ge 0$ and $\Vec{v}(3) \ge 0$. Thus we can safely discard the second coordinate.

    In general, the ``sign-reflecting projection'' developed in Section \ref{sec:srp} shows that any geometrically $2$-dimensional VASS can be projected onto two coordinates so that the signs of these two coordinates reflects the signs of other coordinates.
\end{example}

In the rest of this section we just state formally our main technical results. The detailed proofs are placed in the appendix.

\subsection{Linear Path Schemes}
\label{sec:lps-def}

A \emph{linear path scheme (LPS)} is a pair $(G, \Lambda)$ where $G$ is a VASS and $\Lambda$ is a regular expression of the form
$
    \Lambda = \alpha_0\beta_1^*\alpha_1\ldots\beta_k^*\alpha_k
$
such that $\alpha_0, \ldots, \alpha_k$ are paths in $G$ and $\beta_1, \ldots, \beta_k$ are cycles in $G$, and $\alpha_0\beta_1\alpha_1\ldots \beta_k\alpha_k$ is a path in $G$.
We say an LPS $(G, \Lambda)$ is \emph{compatible to} a VASS $G'$ if $G'$ contains all states and transitions in $\Lambda$.
Very often we shall omit the VASS $G$ and say that $\Lambda$ on its own is an LPS, with $G$ understood as any VASS to which $\Lambda$ is compatible.
We write $|\Lambda| = |\alpha_0\beta_1\alpha_1\ldots\beta_k\alpha_k|$ for the length of $\Lambda$, $\norm{\Lambda} = \norm{\alpha_0\beta_1\alpha_1\ldots\beta_k\alpha_k}$ for its norm, and $|\Lambda|_* = k$ for the number of cycles in $\Lambda$.

We also treat $\Lambda$ as the language defined by it, and thus for two configurations $p(\Vec{u})$ and $q(\Vec{v})$ we write $p(\Vec{u}) \xrightarrow{\Lambda}_{\mathbb{D}} q(\Vec{v})$ if and only if there exists $e_1, \ldots, e_k \in \mathbb{N}$ such that
\begin{equation}
    p(\Vec{u}) \xrightarrow{\alpha_0\beta_1^{e_1}\alpha_1\ldots\beta_k^{e_k}\alpha_k}_{\mathbb{D}} q(\Vec{v}).
\end{equation}


\paragraph*{Positive linear path schemes.}
A \emph{positive LPS} is a regular expression of the form
$\Lambda^+ = \alpha_0\beta_1^+\alpha_1\ldots\beta_k^+\alpha_k$
which is similar to a linear path scheme except that we require each cycle to be used at least once. We write $p(\Vec{u}) \xrightarrow{\Lambda^+}_{\mathbb{D}} q(\Vec{v})$ if and only if there are positive integers $e_1, \ldots, e_k \in \mathbb{N}_{>0}$ such that
\begin{equation}
    p(\Vec{u}) \xrightarrow{\alpha_0\beta_1^{e_1}\alpha_1\ldots\beta_k^{e_k}\alpha_k}_{\mathbb{D}} q(\Vec{v}).
\end{equation}
A path $\pi$ is said to be \emph{admitted by $\Lambda^+$} if $\pi = \alpha_0\beta_1^{e_1}\alpha_1\ldots\beta_k^{e_k}\alpha_k$ for some $e_1, \ldots, e_k \in \mathbb{N}_{>0}$.
We prefer positive LPSes as they can be easily characterized by linear inequality systems (see Section \ref{sec:lps-system} for details). In fact, positive LPSes can be easily obtained from LPSes:

\begin{lemma}
    \label{lem:plps-from-lps}
    For every linear path scheme $\Lambda$ there exists a finite set $S$ of positive linear path schemes compatible to the same VASSes with $\Lambda$, such that $\xrightarrow{\Lambda} {}={} \bigcup_{\Lambda^+\in S}\xrightarrow{\Lambda^+}$ and $|\Lambda^+| \le |\Lambda|$ for every $\Lambda^+\in S$.
\end{lemma}

\begin{proof}
    Suppose $\Lambda = \alpha_0\beta_1^*\alpha_1\ldots \beta_k^*\alpha_k$. For each cycle component $\beta_i^*$ in $\Lambda$ we replace it by either $\beta_i^+$ or an empty word nondeterministically. Let $S$ be the set of all resulting positive LPSes. It is obvious that $S$ satisfies the desired properties.
\end{proof}

\subsection{Main Results}

The main results of this section is stated as follows.

\begin{restatable}{theorem}{geoTwoVassFlat}
    \label{thm:geo-2vass-flat-plps}
    Let $G = (Q, T)$ be a geometrically $2$-dimensional $d$-VASS. For every pair of configurations $p(\Vec{u}), q(\Vec{v}) \in Q\times \mathbb{N}^d$ with $p(\Vec{u}) \xrightarrow{*} q(\Vec{v})$ there exists a positive LPS $\Lambda^+$ compatible to $G$ such that $p(\Vec{u}) \xrightarrow{\Lambda^+} q(\Vec{v})$ and $|\Lambda^+| \le |G|^{O(1)}$.
\end{restatable}
We remark that the big-$O$ term here and elsewhere in the paper hides constant that may depend on the dimension $d$, but does not depend on $G, \Vec{u}, \Vec{v}$ or anything else.

By Lemma \ref{lem:plps-from-lps} we know that positive LPSes can be obtained from LPSes. Thus theorem \ref{thm:geo-2vass-flat-plps} follows easily from the following relaxed theorem, which will be proved in the appendix.
\begin{restatable}{theorem}{flatGeoTwoVASS}
    \label{thm:geo-2vass-flat-lps}
    Let $G = (Q, T)$ be a geometrically $2$-dimensional $d$-VASS. For every pair of configurations $p(\Vec{u}), q(\Vec{v}) \in Q\times \mathbb{N}^d$ with $p(\Vec{u}) \xrightarrow{*} q(\Vec{v})$ there exists an LPS $\Lambda$ compatible to $G$ such that $p(\Vec{u}) \xrightarrow{\Lambda} q(\Vec{v})$ and $|\Lambda| \le |G|^{O(1)}$.
\end{restatable}

\section{Characteristic Systems for Linear Path Schemes}
\label{sec:lps-system}

The property that linear path schemes can be fully characterized by linear inequality systems is exploited in \cite{BFGHM15} to derive the $\mathsf{PSPACE}$ upper bound of the reachability problem in 2-VASSes. Here we recall this linear inequality system and its properties.

We mainly focus on positive linear path schemes. Fix $\Lambda = \alpha_0\beta_1^+\alpha_1\ldots\beta_k^+\alpha_k$ to be a positive LPS from state $p$ to $q$ that is compatible to some $d$-VASS $G = (Q, T)$, where $k = |\Lambda|_*$ is the number of cycles in $\Lambda$.
\begin{definition}[{cf.\ \cite[Lem.\ 14]{BFGHM15}}]
    \label{def:eq-lps}
    The \emph{characteristic system $\LPSSystem{\Lambda}$ of the positive LPS $\Lambda$} is the system of linear inequalities such that a triple $\Vec{h} = (\Vec{u}, \Vec{e}, \Vec{v}) \in \mathbb{N}^d \times \mathbb{N}^k \times \mathbb{N}^d$ satisfies $\LPSSystem{\Lambda}$, written $\Vec{h} \models \LPSSystem{\Lambda}$, if and only if the following conditions hold:

    \begin{enumerate}
        \item for every $i = 1, \ldots, k$, $\Vec{e}(i) \ge 1$;
        \item for every $i = 0, \ldots, k$ and every $j = 1, \ldots, |\alpha_i|$,
        \begin{equation}
            \Vec{u} + \Delta(\alpha_0\beta_1^{\Vec{e}(1)}\alpha_1\ldots \alpha_{i - 1}\beta_i^{\Vec{e}(i)}) + \Delta(\alpha_i[1 \ldots j]) \ge\Vec{0};
        \end{equation}
        \item for every $i = 1, \ldots, k$ and every $j = 1, \ldots, |\beta_i|$,
        \begin{alignat}{3}
            \Vec{u} + \Delta(\alpha_0\beta_1^{\Vec{e}(1)}\alpha_1\ldots &\beta_{i-1}^{\Vec{e}(i-1)}\alpha_{i-1}) &{}+{} \Delta(\beta_i[1 \ldots j]) \ge\Vec{0},\\
            \Vec{u} + \Delta(\alpha_0\beta_1^{\Vec{e}(1)}\alpha_1\ldots &\beta_{i-1}^{\Vec{e}(i-1)}\alpha_{i-1}\beta_i^{\Vec{e}(i) - 1}) &{}+{} \Delta(\beta_i[1 \ldots j]) \ge\Vec{0};
        \end{alignat}
        \item and finally, $\Vec{u} + \Delta(\alpha_0\beta_1^{\Vec{e}(1)}\alpha_1\ldots \beta_k^{\Vec{e}(k)}\alpha_k) = \Vec{v}$.
    \end{enumerate}
\end{definition}

The readers can easily verify that these constraints are indeed linear in terms of $\Vec{u}, \Vec{e}, \Vec{v}$. The next lemma shows that $\LPSSystem{\Lambda}$ indeed captures all runs admitted by $\Lambda$.

\begin{restatable}{lemma}{lemLPSSystem}
    \label{lem:LPS-system}
    Let $G$ be a $d$-VASS and $\Lambda = \alpha_0\beta_1^+\alpha_1\ldots\beta_k^+\alpha_k$ be a positive LPS from state $p$ to $q$ compatible to $G$. Then for every $\Vec{u}, \Vec{v} \in \mathbb{N}^d$, $p(\Vec{u}) \xrightarrow{\Lambda} q(\Vec{v})$ if and only if there exists $\Vec{e} \in \mathbb{N}^k$ such that $(\Vec{u}, \Vec{e}, \Vec{v}) \models \LPSSystem{\Lambda}$. 
    Moreover, for every $\Vec{u}, \Vec{v} \in \mathbb{N}^d$ and every $\Vec{e} \in \mathbb{N}^k$ such that $(\Vec{u}, \Vec{e}, \Vec{v}) \models \LPSSystem{\Lambda}$, we have $p(\Vec{u}) \xrightarrow{\alpha_0\beta_1^{\Vec{e}(1)}\alpha_1\ldots \beta_k^{\Vec{e}(k)}\alpha_k} q(\Vec{v})$.
\end{restatable}

We also need to introduce the homogeneous version of $\LPSSystem{\Lambda}$ for technical reasons.

\begin{definition}
    \label{def:eq-lps-0}
    The \emph{homogeneous characteristic system $\LPSSystemHomo{\Lambda}$ of $\Lambda$} is the system of linear inequalities such that a triple $\Vec{h}_0 = (\Vec{u}_0, \Vec{e}_0, \Vec{v}_0) \in \mathbb{N}^d \times \mathbb{N}^k \times \mathbb{N}^d$ satisfies $\LPSSystemHomo{\Lambda}$, written $\Vec{h}_0 \models \LPSSystemHomo{\Lambda}$, if and only if the following conditions hold:

    \begin{enumerate}
        \item for every $i = 0, \ldots, k$, $\Vec{u}_0 + \Delta(\beta_1) \cdot \Vec{e}_0(1) + \cdots + \Delta(\beta_i) \cdot \Vec{e}_0(i) \ge \Vec{0}$;
        \item $\Vec{u}_0 + \Delta(\beta_1) \cdot \Vec{e}_0(1) + \cdots + \Delta(\beta_k) \cdot \Vec{e}_0(k) = \Vec{v}_0$.
    \end{enumerate}
\end{definition}

\section{The Modified KLMST Decomposition Algorithm}
\label{sec:mod-klmst}

In this section we apply our results of Section \ref{sec:flat-geo-2d} to improve the notoriously hard KLMST decomposition algorithm for VASS reachability. Our narration will base on the work of Leroux and Schmitz \cite{LS19}. For readers familiar with \cite{LS19}, the major modifications are listed below:

\begin{itemize}
    \item The decomposition structure is now a sequence of generalized VASS reachability instances linked by (positive) linear path schemes rather than by single transitions.
    \item We introduce a new ``cleaning'' step that replaces all VASS instances which are geometrically $2$-dimensional by polynomial-length linear path schemes compatible to them.
    \item We do not guarantee the exact preservation of action languages at each decomposition step. Instead, we only preserve a subset of action languages. This is a compromise since linear path schemes capture only the reachability relation but not every possible run. Nonetheless, it is enough for the reachability problem.
\end{itemize}


In this section we focus on the effectiveness and correctness of the modified KLMST decomposition algorithm. Its complexity will be analyzed in Section \ref{sec:complexity}.

\subsection{Linear KLM Sequences}

The underlying decomposition structure in the KLMST algorithm was known as KLM sequences, named after Mayr\cite{Mayr81}, Kosaraju\cite{Kosaraju82}, and Lambert\cite{Lambert92}.

\begin{definition}
    A \emph{KLM tuple} of dimension $d$ is a tuple $\bracket{p(\Vec{x})G q(\Vec{y})}$ where $G = (Q, T)$ is a $d$-VASS and $p(\Vec{x}), q(\Vec{y})\in Q\times \mathbb{N}_\omega^d$ are two (generalized) configurations of $G$.
    A \emph{KLM sequence} of dimension $d$ is a sequence of KLM tuples interleaved by transitions of the form
    \begin{equation}
        \xi = \bracket{p_0(\Vec{x}_0)G_0 q_0(\Vec{y}_0)} t_1 \bracket{p_1(\Vec{x}_1)G_1 q_1(\Vec{y}_1)} t_2 \ldots t_k \bracket{p_k(\Vec{x}_k)G_k q_k(\Vec{y}_k)},
    \end{equation}
    where each tuple $\bracket{p_i(\Vec{x}_i)G_i q_i(\Vec{y}_i)}$ is a KLM tuple of dimension $d$ and each $t_i$ is a transition of the form $(q_{i-1}, \Vec{a}_i, p_i)$ from state $q_{i-1}$ to $p_i$ with effect $\Vec{a}_i \in \mathbb{Z}^d$.
\end{definition}

In this paper we generalize the definition of KLM sequences to allow (positive) linear path schemes to connect KLM tuples.

\begin{definition}
    A \emph{linear KLM sequence} of dimension $d$ is a sequence
    \begin{equation}
        \label{eq:lin-klm-seq-form}
        \xi = \bracket{p_0(\Vec{x}_0)G_0 q_0(\Vec{y}_0)} \Lambda_1 \bracket{p_1(\Vec{x}_1)G_1 q_1(\Vec{y}_1)} \Lambda_2 \ldots \Lambda_k \bracket{p_k(\Vec{x}_k)G_k q_k(\Vec{y}_k)},
    \end{equation}
    where each tuple $\bracket{p_i(\Vec{x}_i)G_i q_i(\Vec{y}_i)}$ is a KLM tuple of dimension $d$ and each $\Lambda_i$ is a positive linear path scheme from state $q_{i-1}$ to $p_i$.
\end{definition}

One immediately sees that KLM sequences are just special cases of linear KLM sequences. Let $\xi$ be a linear KLM sequence given as (\ref{eq:lin-klm-seq-form}), we write $\xi_i := \bracket{p_i(\Vec{x}_i)G_i q_i(\Vec{y}_i)}$ for the $i$th KLM tuple occurring in $\xi$.

\paragraph*{Action languages.}

Let $\xi$ be a linear KLM sequence given as (\ref{eq:lin-klm-seq-form}). We say a path $\pi$ from state $p_0$ to $q_k$ is \emph{admitted by} $\xi$, written $\xi\vdash \pi$, if $\pi$ can be written as $\pi = \pi_0\rho_1\pi_1\ldots \rho_k\pi_k$ where $\pi_i$ is a path from $p_i$ to $q_i$ in $G_i$ for each $i = 0, \ldots, k$, and $\rho_i$ is a path admitted by $\Lambda_i$ for each $i = 1, \ldots, k$, such that there exist vectors $\Vec{m}_0, \Vec{n}_0, \ldots, \Vec{m}_k, \Vec{n}_k \in \mathbb{N}^d$ such that
\begin{equation}
    p_0(\Vec{m}_0) \xrightarrow{\pi_0} q_0(\Vec{n}_0)
    \xrightarrow{\rho_1}
    p_1(\Vec{m}_1) \xrightarrow{\pi_1} q_1(\Vec{n}_1)
    \xrightarrow{\rho_2} \cdots \xrightarrow{\rho_k}
    p_k(\Vec{m}_k) \xrightarrow{\pi_k} q_k(\Vec{n}_k)
\end{equation}
and that $\Vec{m}_i \sqsubseteq \Vec{x}_i$, $\Vec{n}_i \sqsubseteq \Vec{y}_i$ for each $i = 0, \ldots, k$.

The \emph{action language} $L_\xi$ of $\xi$ is the language over $\mathbb{Z}^d$ defined by $L_\xi := \bigl\{ \llbracket \pi \rrbracket : \xi \vdash \pi \bigr\}$,
where we recall that $\llbracket \cdot \rrbracket$ is the word morphism mapping each transition to its effect.

We are more interested in the action languages because in some decomposition steps we have to modify the set of transitions, and only the action word of admitted runs can be preserved. Notice that action languages preserve not only the effects of admitted runs, but also their lengths.

\paragraph*{Ranks and Sizes.}

Let $t$ be a transition in a $d$-VASS $G$, we define $\CycleSpace(G/t)$ to be the vector space spanned by the effects of all cycles in $G$ containing $t$. For the VASS $G$, let $r_i$ be the number of transitions $t$ in $G$ such that $\dim(\CycleSpace(G/t)) = i$ for each $i = 0, \ldots, d$. Then the \emph{rank} of $G$ is defined as $\Rank(G) = (r_d, \ldots, r_3) \in \mathbb{N}^{d-2}$. We also define the \emph{full rank} of $G$ as $\RankFull(G) = (r_d, \ldots, r_0) \in \mathbb{N}^{d+1}$.

The following lemma was proved in \cite{LS19}, which shows that in a strongly connected VASS $G$, the space $\CycleSpace(G/t)$ corresponds to $\CycleSpace(G)$.
\begin{lemma}[{\cite[Lem. 3.2]{LS19}}]
    \label{lem:sc-vass-cycspace}
    Let $t$ be a transition of a strongly connected VASS $G$. Then $\CycleSpace(G/t) = \CycleSpace(G)$.
\end{lemma}

The following corollary is immediate.
\begin{corollary}
    \label{cor:sc-vass-rank0-iff-geo-2d}
    Let $G$ be a strongly connected $d$-VASS. Then $\Rank(G) = \Vec{0}$ if and only if $G$ is geometrically $2$-dimensional.
\end{corollary}

Let $\xi$ be a linear KLM sequence given as (\ref{eq:lin-klm-seq-form}). We define the \emph{rank} of $\xi$ as $\Rank(\xi) = \sum_{i = 0}^k \Rank(G_i)$, and the \emph{full rank} of $\xi$ as $\RankFull(\xi) = \sum_{i = 0}^k \RankFull(G_i)$. We remark that the full rank corresponds to the rank defined in \cite{LS19}. Ranks are ordered lexicographically: let $\Vec{r} = (r_d, \ldots, r_0)$ and $\Vec{r}' = (r_d', \ldots, r_0')$, we write $\Vec{r} \lelex \Vec{r}'$ if $\Vec{r} = \Vec{r}'$ or the maximal $i$ with $r_i \ne r_i'$ satisfies $r_i < r_i'$.

Recall that for a VASS $G$ we write $|G|$ for its size as defined in (\ref{eq:vass-size}). For a linear path scheme $\Lambda$, its length $|\Lambda|$ and norm $\norm{\Lambda}$ are defined in Section \ref{sec:lps-def}.
Let $\zeta = \bracket{p(\Vec{x})Gq(\Vec{y})}$ be a KLM tuple of dimension $d$, its \emph{size} is defined to be $|\zeta| = |G| + d\cdot (\norm{\Vec{x}} + \norm{\Vec{y}} + 1)$.
Let $\xi$ be a linear KLM sequence given as (\ref{eq:lin-klm-seq-form}), we define its \emph{size} as
\begin{equation}
    |\xi| = \sum_{i = 0}^k |\xi_i| + \sum_{i = 1}^{k} d\cdot |\Lambda_i| \cdot (\norm{\Lambda_i} + 1).
\end{equation}

Note that the sizes defined in this paper reflect the sizes of unary encoding, thus have an exponential expansion in their binary encoding.

\subsection{Characteristic Systems for Linear KLM Sequences}

We define in this section the characteristic systems of linear KLM sequences, which are systems of linear inequalities that serve as an under-specification of admitted runs. Let $G = (Q, T)$ be a VASS, we first recall the \emph{Kirchhoff system} $K_{G, p, q}$ of $G$ with respect to states $p, q \in Q$, which is a system of linear equations such that a function $\phi \in \mathbb{N}^{T}$ is a model of $K_{G, p, q}$, written $\phi\models K_{G, p, q}$, if and only if
\begin{equation}
    \mathbf{1}_q - \mathbf{1}_p = \sum_{t = (r, \Vec{a}, s) \in T} \phi(t) \cdot (\mathbf{1}_s - \mathbf{1}_r),
\end{equation}
where $\mathbf{1}_p \in \{0, 1\}^Q$ is the indicator function defined by
$\mathbf{1}_p(q) = 1$ if $q = p$ and $\mathbf{1}_p(q) = 0$ otherwise.
We also need the homogeneous version of $K_{G, p, q}$, denoted by $K_{G, p, q}^0$, where a function $\phi \in \mathbb{N}^T$ is a model of it, written $\phi\models K_{G, p, q}^0$, if and only if
\begin{equation}
    \Vec{0} = \sum_{t = (r, \Vec{a}, s) \in T} \phi(t) \cdot (\mathbf{1}_s - \mathbf{1}_r).
\end{equation}

\begin{definition}
    \label{def:eq-klm}
    Let $\xi$ be a linear KLM sequence given by
    \begin{equation}
        \xi = \bracket{p_0(\Vec{x}_0)G_0q_0(\Vec{y}_0)} \Lambda_1 \bracket{p_1(\Vec{x}_1)G_1q_1(\Vec{y}_1)} \Lambda_2 \ldots \Lambda_k \bracket{p_k(\Vec{x}_k) G_k q_k(\Vec{y}_k)}
    \end{equation}
    The \emph{characteristic system} $\KLMSystem{\xi}$ is a set of linear (in)equalities such that a sequence
    \begin{equation}
        \label{eq:klm-system-solution-form}
        \Vec{h} = (\Vec{m}_0, \phi_0, \Vec{n}_0), \Vec{e}_1, (\Vec{m}_1, \phi_1, \Vec{n}_1), \Vec{e}_2, \ldots, \Vec{e}_k, (\Vec{m}_k, \phi_k, \Vec{n}_k),
    \end{equation}
    where each $(\Vec{m}_i, \phi_i, \Vec{n}_i) \in \mathbb{N}^d \times \mathbb{N}^{T_i} \times \mathbb{N}^d$ and each $\Vec{e}_i \in \mathbb{N}^{|\Lambda_i|_*}$, is a model of $\KLMSystem{\xi}$, written $\Vec{h} \models \KLMSystem{\xi}$, if and only if
    \begin{enumerate}
        \item $\Vec{m}_i \sqsubseteq \Vec{x}_i$, $\phi_i \models K_{G, p, q}$, $\Vec{n}_i \sqsubseteq \Vec{y}_i$ and $\Vec{n}_i = \Vec{m}_i + \Delta(\phi_i)$ for every $i = 0, \ldots, k$;
        \item $(\Vec{n}_{i-1}, \Vec{e}_i, \Vec{m}_i) \models \LPSSystem{\Lambda_i}$ for every $i = 1, \ldots, k$.
    \end{enumerate}

    Similarly, the \emph{homogeneous characteristic system} $\KLMSystemHomo{\xi}$ is a set of linear (in)equalities such that a sequence $\Vec{h}$
    of the form (\ref{eq:klm-system-solution-form}) is a model of $\KLMSystemHomo{\xi}$, written $\Vec{h} \models \KLMSystemHomo{\xi}$, if and only if
    \begin{enumerate}
        \item $\Vec{m}_i(j) = 0$ whenever $\Vec{x}_i(j) \ne \omega$, $\phi_i \models K_{G, p, q}^0$, $\Vec{n}_i(j) = 0$ whenever $\Vec{y}_i(j) \ne \omega$, and $\Vec{n}_i = \Vec{m}_i + \Delta(\phi_i)$ for every $i = 0, \ldots, k$;
        \item $(\Vec{n}_{i-1}, \Vec{e}_i, \Vec{m}_i) \models \LPSSystemHomo{\Lambda_i}$ for every $i = 1, \ldots, k$.
    \end{enumerate}

    The sequence $\xi$ is said to be \emph{satisfiable} if $\KLMSystem{\xi}$ has a model, otherwise it's unsatisfiable.
\end{definition}

Let $\Vec{h}$ be a model of $\KLMSystem{\xi}$ (or $\KLMSystemHomo{\xi}$), we shall write $\Vec{m}_i^{\Vec{h}}, \phi_i^{\Vec{h}}, \Vec{n}_i^{\Vec{h}}, \Vec{e}_i^{\Vec{h}}$ for the values of $\Vec{m}_i, \phi_i, \Vec{n}_i, \Vec{e}_i$ assigned by $\Vec{h}$, respectively. Recall that unsatisfiable linear KLM sequences have empty action languages.

\begin{restatable}[{cf.\ \cite[Lem.\ 3.5]{LS19}}]{lemma}{unsatKLMEmptyActLang}
    \label{lem:unsat-klm-empty-act-lang}
    For any unsatisfiable linear KLM sequence $\xi$, $L_{\xi} = \emptyset$.
\end{restatable}

\subsubsection{Bounds on Bounded Values in \texorpdfstring{$\KLMSystem{\xi}$}{Characteristic System}}

We state here a lemma similar to \cite[Lem.\ 3.7]{LS19}, which upper bounds the bounded values in the characteristic system $\KLMSystem{\xi}$. Its proof can be found in the appendix.

\begin{restatable}{lemma}{klmEquationBoundedBounds}
    \label{lem:klm-eq-bounded-bounds}
    Assume that $\xi=\bracket{p_0(\Vec{x}_0)G_0q_0(\Vec{y}_0)} \Lambda_1 \cdots \Lambda_k \bracket{p_k(\Vec{x}_k) G_k q_k(\Vec{y}_k)}$ is satisfiable. Then for every $0 \leq i \leq k$ we have:
    \begin{itemize}
        \item For every $1 \leq j \leq d$, the set of values $\Vec{m}_i^{\Vec{h}}(j)$ where $\Vec{h}\models\KLMSystem{\xi}$ is unbounded if, and only if, there exists a model $\Vec{h}_0$ of $\KLMSystemHomo{\xi}$ such that $\Vec{m}_i^{\Vec{h}_0}(j) > 0$.
        \item For every $t\in T_i$, the set of values $\phi_i^{\Vec{h}}(t)$ where $\Vec{h}\models\KLMSystem{\xi}$ is unbounded if, and only if, there exists a model $\Vec{h}_0$ of $\KLMSystemHomo{\xi}$ such that $\phi_i^{\Vec{h}}(t) > 0$.
        \item For every $1 \leq j \leq d$, the set of values $\Vec{n}_i^{\Vec{h}}(j)$ where $\Vec{h}\models\KLMSystem{\xi}$ is unbounded if, and only if, there exists a model $\Vec{h}_0$ of $\KLMSystemHomo{\xi}$ such that $\Vec{n}_i^{\Vec{h}_0}(j) > 0$.
    \end{itemize}
    Moreover, every bounded value of $\KLMSystem{\xi}$ is bounded by $(10|\xi|)^{12|\xi|}$.
\end{restatable}


\subsection{Cleaning of Linear KLM Sequences}

In this section we define three conditions that require a linear KLM sequence to be \emph{strongly connected}, \emph{pure}, and \emph{saturated}. Together with the satisfiability condition, they make up the so-called ``clean'' condition of linear KLM sequences. Note that the purity condition is new compared to \cite{LS19}, which requires every geometrically $2$-dimensional VASSes occur in a linear KLM sequence to be replaced by linear path schemes.

\paragraph*{Strongly Connected Sequences}

A linear KLM sequence $\xi = \bracket{p_0(\Vec{x}_0)G_0q_0(\Vec{y}_0)} \Lambda_1 \cdots \Lambda_k \bracket{p_k(\Vec{x}_k) G_k q_k(\Vec{y}_k)}$ is \emph{strongly connected} if all the VASSes $G_0, \ldots, G_k$ are strongly connected (as they are understood as directed graphs). One can easily obtain strongly connected sequences by expanding the strongly connected components of each VASS:

\begin{restatable}[{\cite[Lem.\ 4.2]{LS19}}]{lemma}{decomSCC}
    \label{lem:decom-scc}
    For any linear KLM sequence $\xi$ which is not strongly connected, one can compute in time $\ExpF{|\xi|}$ a finite set $\Xi$ of strongly connected linear KLM sequences such that $L_\xi = \bigcup_{\xi' \in \Xi}L_{\xi'}$ and that $\Rank(\xi') \lelex \Rank(\xi)$ and $|\xi'| \le (2d + 1)|\xi|$ for every $\xi' \in \Xi$.
\end{restatable}

\paragraph*{Pure Sequences}

A KLM tuple $\bracket{p(\Vec{x})Gq(\Vec{y})}$ is called \emph{trivial} if $p(\Vec{x}) = q(\Vec{y})$ and $G$ contains no transition and only a single state $p$. In this case we simply write $\bracket{p(\Vec{x})}$ for this tuple. Note that the action language of a trivial tuple contains exactly the empty word.

A linear KLM sequence $\xi = \bracket{p_0(\Vec{x}_0)G_0q_0(\Vec{y}_0)} \Lambda_1 \cdots \Lambda_k \bracket{p_k(\Vec{x}_k) G_k q_k(\Vec{y}_k)}$ is said to be \emph{pure} if $\xi$ is strongly connected and for every $i = 0, \ldots, k$, $\Rank(G_i) = \Vec{0}$ implies that the tuple $\bracket{p_i(\Vec{x}_i)G_iq_i(\Vec{y}_i)}$ is trivial. By Corollary \ref{cor:sc-vass-rank0-iff-geo-2d}, a rank-$\Vec{0}$ strongly connected VASS is geometrically $2$-dimensional, and thus can be replaced by linear path schemes in case it is not trivial.

\begin{restatable}{lemma}{decomPure}
    \label{lem:decom-pure}
    Let $\xi$ be a strongly connected linear KLM sequence. Whether $\xi$ is pure is in $\mathsf{PSPACE}$. If $\xi$ is not pure, one can compute in space $\PolyF{|\xi|}$ a finite set $\Xi$ of pure linear KLM sequences such that $\bigcup_{\xi'\in\Xi} L_{\xi'} \subseteq L_\xi$ and $\bigcup_{\xi'\in\Xi} L_{\xi'} \ne\emptyset$ whenever $L_\xi \ne \emptyset$, and such that $\Rank(\xi') = \Rank(\xi)$ and $|\xi'| \le |\xi|^{O(1)}$ for all $\xi'\in \Xi$.
\end{restatable}

\paragraph*{Saturated Sequences}

Let $\xi = \bracket{p_0(\Vec{x}_0)G_0q_0(\Vec{y}_0)} \Lambda_1 \cdots \Lambda_k \bracket{p_k(\Vec{x}_k) G_k q_k(\Vec{y}_k)}$ be a linear KLM sequence.
We say $\xi$ is \emph{saturated} if for every $0 \le i \le k$ and every $j \in [d]$, we have
\begin{itemize}
    \item $\Vec{x}_i(j) = \omega$ implies the set of values $\Vec{m}_i^{\Vec{h}}(j)$ where $\Vec{h} \models \KLMSystem{\xi}$ is unbounded; and
    \item $\Vec{y}_i(j) = \omega$ implies the set of values $\Vec{n}_i^{\Vec{h}}(j)$ where $\Vec{h} \models \KLMSystem{\xi}$ is unbounded.
\end{itemize}

\begin{lemma}[{\cite[Lem.\ 4.4]{LS19}}]
    \label{lem:decom-satur}
    From any pure linear KLM sequence $\xi$, one can compute in time $\ExpF{|\xi|^{O(|\xi|)}}$ a finite set $\Xi$ of saturated pure linear KLM sequences such that $L_\xi = \bigcup_{\xi' \in \Xi} L_{\xi'}$, and such that $\Rank(\xi') = \Rank(\xi)$ and $|\xi'| \le |\xi|^{O(|\xi|)}$ for every $\xi' \in \Xi$.
\end{lemma}

\begin{proof}
    By Lemma \ref{lem:klm-eq-bounded-bounds}, if a variable $\Vec{m}_i(j)$ or $\Vec{n}_i(j)$ is bounded in $\KLMSystem{\xi}$, we can replace the corresponding $\omega$ component in $\xi$ by all possible values bounded by $(10|\xi|)^{12|\xi|} \le |\xi|^{O(|\xi|)}$.
\end{proof}

\paragraph*{The Cleaning Lemma}
\label{sec:clean-lemma}

A linear KLM sequence $\xi$ is called \emph{clean} if it is satisfiable, strongly connected, pure and saturated. The lemmas \ref{lem:decom-scc} through \ref{lem:decom-satur} show how to make a linear KLM sequence clean.

\begin{restatable}{lemma}{lemCleaning}
    \label{lem:cleaning}
    From any linear KLM sequence $\xi$, one can compute in time $\ExpF{g(|\xi|)}$ a finite set $\Clean{\xi}$ of clean linear KLM sequences such that $\bigcup_{\xi'\in\Clean{\xi}} L_{\xi'} \subseteq L_\xi$ and $\bigcup_{\xi'\in\Clean{\xi}} L_{\xi'} \ne\emptyset$ whenever $L_\xi \ne \emptyset$. Moreover, for every $\xi' \in \Clean{\xi}$ we have $\Rank(\xi') \lelex \Rank(\xi)$ and $|\xi'| \le g(|\xi|)$ where $g(x) = x^{x^{O(1)}}$.
\end{restatable}

\subsection{Decomposition of Linear KLM Sequences}

In this section we recall three conditions that require a linear KLM sequence to be \emph{unbounded}, \emph{rigid}, and \emph{pumpable}. If any one of them is violated, a decomposition into a set of linear KLM sequences with strictly smaller ranks can be performed. Essentially there is nothing new in this section compared to \cite{LS19}. The decomposition operations in \cite{LS19} can be directly applied here, since they operate on a single KLM tuple and produce KLM sequences that are just special cases of linear KLM sequences. The proofs in \cite{LS19} can also be adapted easily, and we will omit the details here. Especially, the next lemma shows that the arguments of strict decrease in ranks are still valid even though we discard the lower three components of ranks.

\begin{restatable}{lemma}{pureRankDecrease}
    \label{lem:pure-rank-full-decrease-imply-rank-decrease}
    Let $\xi'$ be a pure linear KLM sequence. For any linear KLM sequence $\xi'$ with $\RankFull(\xi') \ltlex \RankFull(\xi)$, we have $\Rank(\xi')\ltlex \Rank(\xi)$.
\end{restatable}

\paragraph*{Unbounded Sequences}

Let $\xi = \bracket{p_0(\Vec{x}_0)G_0q_0(\Vec{y}_0)} \Lambda_1 \cdots \Lambda_k \bracket{p_k(\Vec{x}_k) G_k q_k(\Vec{y}_k)}$ be a linear KLM sequence.
We say $\xi$ is \emph{unbounded} if for all $i = 0, \ldots, k$ and every transition $t \in T_i$ where $T_i$ is the set of transitions of $G_i$, the set of values $\phi_i^{\Vec{h}}(t)$ where $\Vec{h} \models \KLMSystem{\xi}$ is unbounded. Bounded transitions can be expanded exhaustively according to the bounds given by Lemma \ref{lem:klm-eq-bounded-bounds}.

\begin{lemma}[{\cite[Lem.\ 4.6]{LS19}}]
    \label{lem:decom-unb}
    Whether a linear KLM sequence $\xi$ is unbounded is decidable in $\mathsf{NP}$. Moreover, if $\xi$ is pure and bounded, one can compute in time $\ExpF{|\xi|^{O(|\xi|)}}$ a finite set $\Xi$ of linear KLM sequences such that $L_\xi = \bigcup_{\xi' \in \Xi}L_{\xi'}$ and such that $\Rank(\xi') \ltlex \Rank(\xi)$ and $|\xi'| < |\xi|^{O(|\xi|)}$ for every $\xi' \in \Xi$.
\end{lemma}

\paragraph*{Rigid Sequences}

A coordinate $j\in [d]$ is said to be \emph{fixed} by a VASS $G = (Q, T)$ if there exists a function $f_j : Q \to \mathbb{N}$ such that $f_j(q) = f_j(p) + \Vec{a}(j)$ for every transition $(p, \Vec{a}, q) \in T$. We also say that $f_j$ \emph{fixes} $G$ at coordinate $j$ in this case.

A KLM tuple $\bracket{p(\Vec{x}) G q(\Vec{y})}$ is said to be \emph{rigid} if for every coordinate $j$ fixed by $G = (Q, T)$, there exists a function $g_j : Q \to \mathbb{N}$ that fixes $G$ at coordinate $j$ and such that $g_j(p) \sqsubseteq \Vec{x}(j)$ and $g_j(q) \sqsubseteq \Vec{y}(j)$.

A linear KLM sequence $\xi = \bracket{p_0(\Vec{x}_0)G_0q_0(\Vec{y}_0)} \Lambda_1 \cdots \Lambda_k \bracket{p_k(\Vec{x}_k) G_k q_k(\Vec{y}_k)}$
is said to be \emph{rigid} if every tuple $\bracket{p_i(\Vec{x}_i) G_i q_i(\Vec{y}_i)}$ in $\xi$ is rigid.

\begin{lemma}[{\cite[Lem.\ 4.9]{LS19}}]
    \label{lem:decom-rigid}
    From any pure linear KLM sequence $\xi$ one can decide in polynomial time whether $\xi$ is not rigid. Moreover, in that case one can compute in polynomial time a linear KLM sequence $\xi'$ such that $L_{\xi} = L_{\xi'}$, $\Rank( \xi' ) \ltlex \Rank( \xi )$, and $|\xi'| \le |\xi|$.
\end{lemma}

\paragraph*{Pumpable Sequences}

Given a KLM tuple $\bracket{p(\Vec{x})Gq(\Vec{y})}$, recall the \emph{forward and backward acceleration vectors} $\Facc_{G, p}(\Vec{x}), \Bacc_{G, q}(\Vec{y}) \in \mathbb{N}_\omega^d$ defined by
\begin{align}
    \Facc_{G, p}(\Vec{x})(j) &= \begin{cases}
        \omega & \text{if } p(\Vec{x}) \xrightarrow{*} p(\Vec{x}') \text{ for some } \Vec{x}' \text{ with } \Vec{x}' \ge \Vec{x}, \Vec{x}'(j) > \Vec{x}(j)\\
        \Vec{x}(j) & \text{otherwise}
    \end{cases}\\
    \Bacc_{G, q}(\Vec{y})(j) &= \begin{cases}
        \omega & \text{if } q(\Vec{y}') \xrightarrow{*} q(\Vec{y}) \text{ for some } \Vec{y}' \text{ with } \Vec{y}' \ge \Vec{y}, \Vec{y}'(j) > \Vec{y}(j)\\
        \Vec{y}(j) & \text{otherwise}
    \end{cases}
\end{align}

A tuple $\bracket{p(\Vec{x}) G q(\Vec{y})}$ is said to be \emph{pumpable} if $\Facc_{G, p}(\Vec{x})(j) = \Bacc_{G, q}(\Vec{y})(j) = \omega$ for every coordinate $j$ \emph{not} fixed by $G$.

A linear KLM sequence given by $\xi = \bracket{p_0(\Vec{x}_0)G_0q_0(\Vec{y}_0)} \Lambda_1 \cdots \Lambda_k \bracket{p_k(\Vec{x}_k) G_k q_k(\Vec{y}_k)}$
is said to be \emph{pumpable} if every tuple $\bracket{p_i(\Vec{x}_i) G_i q_i(\Vec{y}_i)}$ in $\xi$ is pumpable.

\begin{lemma}[{\cite[Lem.\ 4.15]{LS19}}]
    \label{lem:decom-pump}
    Whether a linear KLM sequence $\xi$ is pumpable is decidable in $\mathsf{EXPSPACE}$. Moreover, if $\xi$ is pure and unpumpable, one can compute in time $\ExpF{|\xi|^{O(1)}}$ a finite set $\Xi$ of linear KLM sequences such that $L_\xi = \bigcup_{\xi' \in \Xi}L_{\xi'}$ and such that $\Rank(\xi') \ltlex \Rank(\xi)$ and $|\xi'| < |\xi|^{O(1)}$ for every $\xi' \in \Xi$.
\end{lemma}

Note that the $O(1)$ term here hides a constant depending on $d$, which essentially arises from a result on the coverability problem by Rackoff \cite{Rackoff78}. The $O(1)$ term also captures the difference between the sizes of linear KLM sequences defined here and that in \cite{LS19}.

\paragraph*{The Decomposition Lemma}

A linear KLM sequence is \emph{normal} if it is clean, unbounded, rigid, and pumpable. The lemmas \ref{lem:decom-unb} through \ref{lem:decom-pump} show that when a clean linear KLM sequence is not normal, we are able to decompose it into a finite set of linear KLM sequences of strictly smaller ranks.

\begin{restatable}{lemma}{klmSeqDecomposition}
    \label{lem:klm-seq-decomposition}
    From any clean linear KLM sequences $\xi$ that is not normal, one can compute in time $\ExpF{h(\xi)}$ a finite set $\Dec{\xi}$ of clean linear KLM sequence such that $\bigcup_{\xi'\in\Dec{\xi}} L_{\xi'} \subseteq L_\xi$ and $\bigcup_{\xi'\in\Dec{\xi}} L_{\xi'} \ne\emptyset$ whenever $L_\xi \ne \emptyset$. Moreover, for every $\xi' \in \Dec{\xi}$ we have $\Rank(\xi') \ltlex \Rank(\xi)$ and $|\xi'| \le h(\xi)$ where $h(x) = x^{x^{x^{O(1)}}}$.
\end{restatable}

\subsection{Normal Sequences}

The following lemma shows that a normal linear KLM sequence is guaranteed to have non-empty action language, thus one can terminate the decomposition process once a normal sequence is produced.

\begin{restatable}{lemma}{klmNormalBoundedWitness}
    \label{lem:klm-normal-bounded-witness}
    Let $\xi$ be a normal linear KLM sequence, then there is a word $\sigma \in L_\xi$ whose length is bounded by $|\sigma| \le \ell(|\xi|)$ where $\ell(x) \le x^{O(x)}$.
\end{restatable}

\subsection{Putting All Together: The Modified KLMST Algorithm}

Here we describe the modified KLMST decomposition algorithm for VASS reachability problem. Suppose we are given a $d$-VASS $G = (Q, T)$ and two configurations $p(\Vec{m}), q(\Vec{n}) \in Q\times \mathbb{N}^d$. To decide whether $p(\Vec{m}) \xrightarrow{*} q(\Vec{n})$ holds in $G$, it is enough to decide whether $L_\xi$ is non-empty where $\xi = \bracket{p(\Vec{m}) G q(\Vec{n})}$. To start with, we use Lemma \ref{lem:cleaning} to clean the sequence $\xi$, and then choose $\xi^0 \in \Clean{\xi}$ non-deterministically. If $\xi^0$ is normal then we are done by Lemma \ref{lem:klm-normal-bounded-witness}. Otherwise, we decompose $\xi^0$ using Lemma \ref{lem:klm-seq-decomposition} and choose $\xi^1 \in \Dec{\xi^0}$ non-deterministically. The procedure continues to produce a series of linear KLM sequences $\xi^0, \xi^1, \xi^2, \ldots$ where $\xi^{i + 1} \in \Dec{\xi^i}$, until either we finally obtain a normal sequence $\xi^n$, or at some point we have to abort because the decomposition of a linear KLM sequence is the empty set. The procedure terminates because $\Rank(\xi^0) >_{\text{lex}} \Rank(\xi^1) >_{\text{lex}} \Rank(\xi^2) >_{\text{lex}} \cdots$ form a decreasing chain of the well-order $(\mathbb{N}^{d-2}, \ltlex)$, which must be finite. If $L_\xi = \emptyset$ then we cannot get a normal sequence since the action languages $L_\xi \supseteq L_{\xi^0} \supseteq L_{\xi^1} \supseteq \cdots$ are all empty. On the other hand, if $L_\xi \ne\emptyset$ then there are non-deterministic choices that always choose the linear KLM sequences with non-empty action languages, which finally lead to a normal sequence. This shows the correctness of the algorithm.

\section{Complexity Upper Bound}
\label{sec:complexity}

The termination of the modified KLMST decomposition algorithm is guaranteed by a ranking function that decreases along a well-ordering. In order to analyze the length of this decreasing chain, we recall the so-called ``length function theorems'' by Schmitz \cite{Schmitz14} in Section \ref{sec:length-func-thm}. After that, we can locate the complexity upper bound of the algorithm in the fast-growing complexity hierarchy \cite{Schmitz16} which we recall in Section \ref{sec:fast-grow-hierar}. Readers familiar with \cite{LS19} may realize that the complexity upper bound for $d$-VASS can be improved to $\mathsf{F}_{d + 1}$, i.e.\ the $(d{ +}1)$-th level in the fast-growing hierarchy, with our ranking function. In fact, by a careful analysis on a property of fast-growing functions, we further improve this bound to $\mathsf{F}_d$.

In this section we assume some familiarity with ordinal numbers (see, e.g.\ \cite{Jech03}). We write $\omega$ here for the first infinite ordinal, not to be confused with the infinite element in previous sections.

\subsection{Length of Sequences of Decreasing Ranks}
\label{sec:length-func-thm}


Let $\xi$ be a linear KLM sequence of dimension $d$ with $\Rank(\xi) = (r_d, \ldots, r_3)$, we define the \emph{ordinal rank} $\alpha_\xi$ of $\xi$ as the ordinal number given by
\begin{equation}
    \alpha_\xi := \omega^{d - 3} \cdot r_d + \omega^{d - 4}\cdot r_{d - 1} + \cdots + \omega^0 \cdot r_3.
\end{equation}
Note that $\Rank(\xi) \ltlex \Rank(\xi')$ if and only if $\alpha_\xi < \alpha_{\xi'}$. With this reformulation, we now focus on the decreasing chain of ordinal ranks.

Let $\alpha < \omega^{\omega}$ be an ordinal given in Cantor Normal Form as $\alpha = \omega^n \cdot c_n + \cdots + \omega^0 \cdot c_0$ where $n, c_0, \ldots, c_n \in \mathbb{N}$, we define the \emph{size} of $\alpha$ as $N\alpha := \max\{n, \max_{0\le i\le n}c_i\}$. For the linear KLM sequence $\xi$ with $\Rank(\xi) = (r_d, \ldots, r_3)$, we have $N\alpha_{\xi} = \max\{d - 3, \max_{3\le i\le d}r_i\} \le |\xi|$.

Given a number $n_0 \in \mathbb{N}$ and a function $h: \mathbb{N} \to \mathbb{N}$ that is monotone inflationary (that is, $x \le h(x)$ and $h(x) \le h(y)$ whenever $x\le y$), we say a sequence of ordinals $\alpha_0, \alpha_1, \ldots$ is \emph{$(n_0, h)$-controlled} if $N\alpha_i \le h^i(n_0)$ for all $i \in \mathbb{N}$, where $h^i(n_0)$ is the $i$th iteration of $h$ on $n_0$.

Let $\xi^0, \xi^1, \ldots$ be the linear KLM sequences produced in the modified KLMST algorithm, by Lemma \ref{lem:klm-seq-decomposition} we know that the sequence of ordinal ranks
\begin{equation}
    \label{eq:decrease-ordinal-ranks}
    \alpha_{\xi^0} > \alpha_{\xi^1} > \cdots
\end{equation}
is $(|\xi^0|, h)$-controlled where $h$ is defined in Lemma \ref{lem:klm-seq-decomposition}. Recall that $\xi^0 \in \Clean{\xi}$ where $\xi = \bracket{p(\Vec{m}) G q(\Vec{n})}$ corresponds to the input reachability instance. Then $|\xi^0| \le g(|\xi|)$ where $g$ is defined in Lemma \ref{lem:cleaning}, and (\ref{eq:decrease-ordinal-ranks}) is indeed $(g(n), h)$-controlled where $n := |\bracket{p(\Vec{m}) G q(\Vec{n})}|$.

\paragraph*{Length function theorem.}

The length of the controlled sequence of ordinals (\ref{eq:decrease-ordinal-ranks}) can be bounded in terms of the hierarchies of fast-growing functions of Hardy and Cicho\'n \cite{CB98}. First recall that given a limit ordinal $\lambda \le \omega^\omega$, the standard \emph{fundamental sequence} of $\lambda$ is a sequence $(\lambda(x))_{x<\omega}$ defined inductively by
\begin{equation}
    \omega^\omega(x) := \omega^{x + 1},
    \qquad
    (\beta + \omega^{k + 1})(x) := \beta + \omega^{k} \cdot (x + 1)
\end{equation}
where $\beta + \omega^{k+1}$ is in Cantor Normal Form. Now given a function $h : \mathbb{N} \to \mathbb{N}$ that is monotone inflationary, we define the \emph{Hardy hierarchy} $(h^\alpha)_{\alpha\le \omega^\omega}$ and the \emph{Cicho\'n hierarchy} $(h_\alpha)_{\alpha\le \omega^\omega}$ as two families of functions $h^\alpha, h_\alpha: \mathbb{N}\to \mathbb{N}$ indexed by ordinals $\alpha \le \omega^\omega$ given inductively by
\begin{alignat}{3}
    h^0(x) &:= x, \qquad
    h^{\alpha+1}(x) :={}& h^\alpha(h(x)), \qquad
    h^\lambda(x) &:= h^{\lambda(x)}(x), \\
    h_0(x) &:= 0, \qquad
    h_{\alpha+1}(x) :={}& 1 + h_{\alpha}(h(x)), \qquad
    h_\lambda(x) &:= h_{\lambda(x)}(x).
\end{alignat}
Observe that Cicho\'n hierarchy counts the number of iterations of $h$ in Hardy hierarchy, that is, $h^{h_\alpha(x)}(x) = h^\alpha(x)$.
Also note that as $h$ is monotone inflationary, by induction on $\alpha$ we have $h_{\alpha}(x) \le h^{\alpha}(x)$.
Now we give the length function theorem as follows.
\begin{theorem}[Length function theorem, {\cite[Thm.\ 3.3]{Schmitz14}}]
    \label{thm:length-func-thm}
    Let $n_0 \ge d - 2$, then the maximal length of $(n_0, h)$-controlled decreasing sequences of ordinals in $\omega^{d-2}$ is $h_{\omega^{d-2}}(n_0)$.
\end{theorem}

\paragraph*{Small witness property.}

By Theorem \ref{thm:length-func-thm} we can bound the length of (\ref{eq:decrease-ordinal-ranks}), which then yields a bound on the minimal length of runs witnessing reachability.

\begin{restatable}[Small witnesses]{lemma}{lemSmallWitness}
    \label{lem:small-witness}
    Let $G = (Q, T)$ be a $d$-VASS where $d \ge 3$, let $p(\Vec{m}), q(\Vec{n}) \in Q\times \mathbb{N}^d$ be two configurations, and let $n := |\bracket{p(\Vec{m}) G q(\Vec{n})}|$. If $p(\Vec{m}) \xrightarrow{*} q(\Vec{n})$ holds in $G$, then there is a path $\sigma$ such that $p(\Vec{m}) \xrightarrow{\sigma} q(\Vec{n})$ and
    $|\sigma| \le \ell(h^{\omega^{d-2}}(g(n)))$,
    where $g, h, \ell$ are defined in lemmas \ref{lem:cleaning}, \ref{lem:klm-seq-decomposition}, and \ref{lem:klm-normal-bounded-witness}.
\end{restatable}

\begin{proof}
    Suppose $p(\Vec{m}) \xrightarrow{*} q(\Vec{n})$, then there is a sequence of linear KLM sequence $\xi^0, \xi^1, \ldots, \xi^L$ produced in the modified KLMST algorithm, such that $\xi^L$ is normal. We have discussed that the sequence of their ordinal ranks $\alpha_{\xi^0} > \alpha_{\xi^1} > \cdots > \alpha_{\xi^L}$ is $(g(n), h)$-controlled, so by Theorem \ref{thm:length-func-thm} we have $L \le h_{\omega^{d-2}}(g(n))$. From Lemma \ref{lem:klm-seq-decomposition} and the fact that $h^{h_\alpha(x)} = h^\alpha(x)$, the size of $\xi^L$ is bounded by
    \begin{equation}
        |\xi^L| \le h^{L}(|\xi^0|) \le h^{h_{\omega^{d-2}}(g(n))}(g(n)) = h^{\omega^{d-2}}(g(n)).
    \end{equation}
    Now Lemma \ref{lem:klm-normal-bounded-witness} bounds the length of the minimal witnesses by $\ell(h^{\omega^{d-2}}(g(n)))$.
\end{proof}

\subsection{Fast-Growing Complexity Hierarchy}
\label{sec:fast-grow-hierar}

We recall the fast-growing hierarchy formally introduced by Schmitz \cite{Schmitz16} that captures the complexity class high above elementary. Define $H(x) := x + 1$, we shall use the Hardy hierarchy $(H^\alpha)_{\alpha < \omega^\omega}$, where for example $H^{\omega^2}(x) = 2^{x+1}(x+1)$ and $H^{\omega^3}(x)$ grows faster than the tower function. First we define the family
$\mathscr{F}_\alpha := \bigcup_{\beta < \omega^\alpha} \mathsf{FDTIME}(H^\beta(n))$
which contains functions computable in deterministic time $O(H^\beta(n))$. Observe that, for example, $\mathscr{F}_3$ contains exactly the Kalmar elementary functions. Now we define
\begin{equation}
    \mathsf{F}_\alpha := \bigcup_{p \in \mathscr{F}_\alpha} \mathsf{DTIME}(H^{\omega^\alpha}(p(n)))
\end{equation}
which is the class of decision problems solvable in deterministic time $O(H^{\omega^\alpha}(p(n)))$. Note that non-deterministic time Turing machines can be made deterministic with an exponential overhead in $\mathscr{F}_3$, thus for $\alpha \ge 3$, we have equivalently that
$\mathsf{F}_\alpha = \bigcup_{p \in \mathscr{F}_\alpha} \mathsf{NDTIME}(H^{\omega^\alpha}(p(n)))$.
Observe that $\mathsf{F}_\alpha$ is closed under reductions in $\mathscr{F}_\alpha$.

\subsubsection{Relativized Fast-Growing Functions}

In order to express the complexity of the modified KLMST algorithm in terms of the hierarchy $(\mathsf{F}_\alpha)_{\alpha < \omega}$, one needs to locate the function $h^{\omega^{d - 2}}$ in the Hardy hierarchy $(H^\alpha)_{\alpha < \omega^\omega}$ where $h \in \mathscr{F}_3$ is the elementary function from Lemma \ref{lem:klm-seq-decomposition}. Previously we can upper bound $h^{\omega^{d - 2}}$ by $H^{\omega^{d + 1}}$ with the help of \cite[Lem.\ 4.2]{Schmitz16}. Here we show a slightly better result, from which we can bound $h^{\omega^{d - 2}}(x)$ by $H^{\omega^d}(O(x))$.

\begin{restatable}[{cf.\ \cite[Lem.\ A.5]{Schmitz16}}]{lemma}{relFastGrowFunc}
    \label{lem:rel-fast-grow-func}
    Let $h: \mathbb{N}\to \mathbb{N}$ be a monotone inflationary function, let $a, b, c\ge 1$ and $x_0 \ge 0$ be natural numbers. If for all $x \ge x_0$, $h(x) \le H^{\omega^b \cdot c}(x)$, then $h^{\omega^a}(x) \le H^{\omega^{b + a}}((c + 1)x)$ for all $x \ge \max\{2c, x_0\}$.
\end{restatable}

\subsection{Upper Bounds for VASS Reachability}

Now we analyze the time complexity of the modified KLMST algorithm. Given as input the $d$-VASS $G = (Q, T)$ and two configurations $p(\Vec{m}), q(\Vec{n})$, let $\xi := \bracket{p(\Vec{m})Gq(\Vec{n})}$ and $n := |\xi|$. The initial sequence $\xi^0 \in \Clean{\xi}$ can be computed in (non-deterministic) time elementary in $n$ by Lemma \ref{lem:cleaning}. Then the algorithm produces $\xi^0, \xi^1, \ldots, \xi^L$ with $L \le h_{\omega^{d - 2}}(g(n))$, where $g, h$ are defined in lemmas \ref{lem:cleaning}, \ref{lem:klm-seq-decomposition}. Note that in each step, the sequence $\xi^{i + 1} \in \Dec{\xi^i}$ can be computed in time elementary in $|\xi^i|$ by Lemma \ref{lem:klm-seq-decomposition}, and the sizes $|\xi^i|$ are all bounded by $h^{\omega^{d-2}}(g(n))$ as we have discussed above in the proof of Lemma \ref{lem:small-witness}. To sum up, the entire algorithm finishes in non-deterministic time elementary in $h^{\omega^{d-2}}(g(n))$.

\begin{lemma}
    \label{lem:mod-klmst-time}
    On input a $d$-VASS $G = (Q, T)$ where $d\ge 3$ and $p(\Vec{m}), q(\Vec{n}) \in Q\times\mathbb{N}^d$, the modified KLMST algorithm finishes in non-deterministic time
    $e(h^{\omega^{d-2}}(g(n)))$ where $n = |\langle {p(\Vec{m})Gq(\Vec{n})} \rangle|$, g, h are defined in lemmas \ref{lem:cleaning}, \ref{lem:klm-seq-decomposition}, and $e \in \mathscr{F}_3$ is some fixed function.
\end{lemma}

Since $h$ is an elementary function, there is a number $c \in \mathbb{N}$ such that $h$ is eventually dominated by $H^{\omega^2\cdot c}$. By Lemma \ref{lem:rel-fast-grow-func} we can upper bound $h^{\omega^{d-2}}(x)$ by $H^{\omega^d}((c + 1)x)$. Observe that the inner part $g(n)$ is elementary in the binary encoding size of the input $G, p(\Vec{m}), q(\Vec{n})$, thus can be captured by a function $p \in \mathscr{F}_3$. Finally, \cite[Lem.\ 4.6]{Schmitz16} allows us to move the outermost function $e$ to the innermost position. Hence we have the following upper bound.

\begin{theorem}
    \label{thm:vass-Fd}
    Reachability in $d$-dimensional VASS is in $\mathsf{F}_d$ for all $d\ge 3$.
\end{theorem}

Also, by Lemma \ref{lem:small-witness} there is a simple combinatorial algorithm for $d$-VASS reachability. We fist compute the bound $B := \ell(h^{\omega^{d-2}}(g(n)))$, which can be done in time elementary in $B$ by \cite[Thm.\ 5.1]{Schmitz16}. Then we can decide reachability by just enumerate all possible paths in $G$ with length bounded by $B$.

\section{Conclusion}
\label{sec:conclusion}

We have shown that the reachability problem in $d$-dimensional vector addition system with states is in $\mathsf{F}_d$, improving the previous $\mathsf{F}_{d + 4}$ upper bound by Leroux and Schmitz \cite{LS19}. By capturing reachability in geometrically $2$-dimensional VASSes with linear path schemes, we are able to reduce significantly the number of decomposition steps in the KLMST decomposition algorithm. Combined with a careful analysis on fast-growing functions, we finally obtained the $\mathsf{F}_d$ upper bound. It should be noticed though, that our algorithm avoids computing the ``full decomposition'' \cite{LS15} of KLM sequences, thus cannot improve the complexity of problems that essentially rely on the full decomposition, e.g., the VASS downward language inclusion problem \cite{HMW10,Zetzsche16,LS19}.

It has been shown that the reachability problem in $(2d{+}3)$-VASS is $\mathsf{F}_d$-hard \cite{CJLLO23}.
In the case of $3$-VASS, it is known that the reachability problem is $\mathsf{PSPACE}$-hard. The gap between the lower bound and the upper bound $\mathsf{F}_3=\mathsf{TOWER}$~\cite{YF23} is huge.
It is very unlikely that the problem is $\mathsf{PSPACE}$-complete.
Effort to uplift the lower bound is called for. 

\bibliography{ref.bib}

\newpage
\appendix

\section{Proof of Theorem \ref{thm:geo-2vass-flat-lps}}

In this appendix we prove the main result of Section \ref{sec:flat-geo-2d}.

\flatGeoTwoVASS*

First note that for a geometrically $2$-dimensional $d$-VASS $G$, we can safely assume that $\dim(\CycleSpace(G)) = 2$ in the proof of Theorem \ref{thm:geo-2vass-flat-lps}. If it is the case that $\dim(\CycleSpace(G)) < 2$, we just add an isolated state $q_\bot$ with self-loops on it to enlarge $\CycleSpace(G)$. Note that $q_\bot$ will never be used in any LPS capturing reachability relation in the original $G$. 

We will prove Theorem \ref{thm:geo-2vass-flat-lps} by induction on $d$. The base case $d = 2$ is already studied in \cite{BEF+21} and will be reviewed in the section \ref{sec:review-flat-2vass}. First let's recall a special kind of linear path schemes, called \emph{zigzag-free} LPSes.

\subsection{Zigzag-free Linear Path Schemes}

An \emph{orthant} is one of the $2^d$ regions in $\mathbb{Q}^d$ split by the $d$ coordinate axes. Formally speaking, given a tuple $t \in \{+1, -1\}^d$, the orthant $Z_t$ defined by $t$ is the following set:
\begin{equation}
    Z_t := \{ \Vec{u} \in \mathbb{Q}^d : \Vec{u}(i) \cdot t(i) \ge 0 \text{ for all } i \in [d]\}.
\end{equation}

An LPS $\Lambda = \alpha_0\beta_1^*\alpha_1\ldots\beta_k^*\alpha_k$ is said to be \emph{zigzag-free in the orthant $Z_t$} if $\Delta(\beta_j) \in Z_t$ for all $j = 1, \ldots, k$. 
We say $\Lambda$ is \emph{zigzag-free} if it is zigzag-free in some orthant.
Zigzag free linear path schemes can cast $\mathbb{Z}^d$-reachability to $\mathbb{N}^d$-reachability when the start and end configurations are sufficiently high:

\begin{lemma}[{\cite[lemma 4.7]{BEF+21} and \cite[lemma 4.6]{LS04}}]
    \label{lem:zigzag-free-LPS-high-reach}
    Let $(G, \Lambda)$ be a zigzag-free LPS where $G = (Q, T)$ is a $d$-VASS. Let $p(\Vec{u}), q(\Vec{v}) \in Q\times \mathbb{N}^d$ be two configurations of $G$.  If $p(\Vec{u}) \xrightarrow{\Lambda}_{\mathbb{Z}^d} q(\Vec{v})$ and $\Vec{u}(i), \Vec{v}(i) \ge |\Lambda| \cdot \norm{T}$ for all $i \in [d]$, then $p(\Vec{u}) \xrightarrow{\Lambda}_{\mathbb{N}^d} q(\Vec{v})$.
\end{lemma}

\subsection{Review of Known Results in Dimension 2}
\label{sec:review-flat-2vass}

The following is the main result of \cite{BEF+21} and is the base case of Theorem \ref{thm:geo-2vass-flat-lps}:

\begin{theorem}[{\cite[Thm.\ 3.1]{BEF+21}}]
    \label{thm:2vass-flat-lps}
    Let $G = (Q, T)$ be a $2$-VASS. For every pair of configurations $p(\Vec{u}), q(\Vec{v}) \in Q\times \mathbb{N}^2$ with $p(\Vec{u}) \xrightarrow{*} q(\Vec{v})$ there exists an LPS $\Lambda$ compatible to $G$ such that $p(\Vec{u}) \xrightarrow{\Lambda} q(\Vec{v})$ and $|\Lambda| \le |G|^{O(1)}$.
\end{theorem}

We also need some other lemmas from \cite{BEF+21} to handle cyclic runs far away from axes.

\begin{lemma}[{\cite[Prop.\ 4.2]{BEF+21}}]
    \label{lem:flat-2d-high-cycle}
    For every $2$-VASS $G = (Q, T)$ there exists a natural number $D\le |G|^{O(1)}$ such that for any state $q\in Q$ and any $\Vec{u}, \Vec{v} \in [D, \infty)^2$ with $q(\Vec{u}) \xrightarrow{*}_{\mathbb{N}^2} q(\Vec{v})$ there exists an LPS $\Lambda$ compatible to $G$ such that $q(\Vec{u}) \xrightarrow{\Lambda}_{\mathbb{N}^2} q(\Vec{v})$ and $|\Lambda| \le |G|^{O(1)}$.
\end{lemma}

We can further require that the above LPS $\Lambda$ is zigzag-free.

\begin{lemma}[{\cite[Lem.\ 4.15]{BEF+21}}]
    \label{lem:zigzag-free-cycle-lps}
    For every 2-VASS $G=(Q, T)$, every $q \in Q$, and every LPS $\sigma$ from $q$ to $q$ compatible to $G$, there exists a finite set $S$ of zigzag-free LPSes from $q$ to $q$ compatible to $G$ such that $\Delta(\sigma) \subseteq \Delta(S)$, and $|\Lambda| \le (|\sigma|+\norm{T})^{O(1)}$ for every $\Lambda \in S$.
\end{lemma}

If the effects of cycles in a zigzag-free LPS $\Lambda$ capturing $q(\Vec{u}) \xrightarrow{*} q(\Vec{v})$ does not belong to any orthant containing $\Vec{v} - \Vec{u}$, they can only be fired for a bounded number of times. We can remove those cycles by expanding them.

\begin{lemma}
    \label{lem:zigzag-cycle-lps-orthant}
    Let $G=(Q, T)$ be a $2$-VASS, let $q \in Q$ and $\Vec{u}, \Vec{v} \in \mathbb{N}^2$. Let $Z$ be any orthant containing $\Vec{v} - \Vec{u}$. For every zigzag-free linear path scheme $\sigma$ compatible to $G$ such that $q(\Vec{u}) \xrightarrow{\sigma} q(\Vec{v})$, there exists an LPS $\Lambda$ from $q$ to $q$ compatible to $G$ that is zigzag-free in $Z$, and such that $q(\Vec{u}) \xrightarrow{\Lambda} q(\Vec{v})$, and $|\Lambda| \leq(|\sigma|+\norm{T})^{O(1)}$.
\end{lemma}

\begin{proof}
    Let $\sigma = \alpha_0\beta_1^*\alpha_1\ldots \beta_k^*\alpha_k$ be a zigzag-free LPS such that $q(\Vec{u}) \xrightarrow{\sigma} q(\Vec{v})$. There exists numbers $e_1, \ldots, e_k \in \mathbb{N}$ such that 
    $ q(\Vec{u}) \xrightarrow{\alpha_0\beta_1^{e_1}\alpha_1\ldots \beta_k^{e_k}\alpha_k} q(\Vec{v}) $.
    In particular, 
    \begin{equation}
        \label{eq:effect-zigzag-free-lps-cycle}
        \Vec{v} - \Vec{u} = \Delta(\alpha_0 \ldots \alpha_k) + \sum_{i = 1}^k e_i \cdot \Delta(\beta_i).
    \end{equation}
    Let $Z$ be an orthant containing $\Vec{v} - \Vec{u}$. Assume w.l.o.g.\ that $Z = \mathbb{Q}^2$ is the first orthant. We claim that for all $j$ such that $\Delta(\beta_j) \notin Z$, $e_j \le |\sigma|\cdot \norm{T}$. Indeed, since $\Delta(\beta_j)\notin Z$ and $\Vec{v} - \Vec{u} \in Z$, there must be an index $\ell\in \{1, 2\}$ such that $(\Vec{v} - \Vec{u})(\ell) \cdot \Delta(\beta_j)(\ell) \le 0$ and $\Delta(\beta_j)(\ell)\ne 0$. We assume w.l.o.g.\ that $(\Vec{v} - \Vec{u})(\ell) \ge 0$ and $\Delta(\beta_j)(\ell) < 0$. Then $\Delta(\beta_i)(\ell)\le 0$ for all $i\in [k]$ by the zigzag-free property of $\sigma$. In particular, $\Delta(\beta_j)(\ell) \le -1$. From the $\ell$-th component of equation (\ref{eq:effect-zigzag-free-lps-cycle}),
    \begin{equation}
        0\le (\Vec{v} - \Vec{u})(\ell) = \Delta(\alpha_0 \ldots \alpha_k)(\ell) + \sum_{i = 1}^k e_i \cdot \Delta(\beta_i)(\ell) \le |\sigma|\cdot \norm{T} - e_j
    \end{equation}
    we have $e_j \le |\sigma|\cdot \norm{T}$ as we claimed. The proof is completed by taking $\Lambda$ to be the LPS obtained from $\sigma$ by replacing all cycles $\beta_j^*$ where $\Delta(\beta_j)\notin Z$ with the path $\beta_j^{e_j}$.
\end{proof}

We combine Lemmas \ref{lem:flat-2d-high-cycle}, \ref{lem:zigzag-free-cycle-lps}, \ref{lem:zigzag-cycle-lps-orthant} to get the following one:

\begin{corollary}
    \label{cor:zigzag-free-in-Z-cycle-lps}
    For every $2$-VASS $G = (Q, T)$ there exists a natural number $D \le |G|^{O(1)}$ such that, for any state $q\in Q$, any $\Vec{u}, \Vec{v} \in [D, \infty)^2$ with $q(\Vec{u}) \xrightarrow{*}_{\mathbb{N}^2} q(\Vec{v})$ and for any orthant $Z$ containing $\Vec{v} - \Vec{u}$, there exists an LPS $\Lambda$ compatible to $G$ that is zigzag-free in $Z$ such that $q(\Vec{u}) \xrightarrow{\Lambda}_{\mathbb{N}^2} q(\Vec{v})$ and $|\Lambda| \le |G|^{O(1)}$.
\end{corollary}

\begin{proof}
    Let $D_0$ be the parameter $D$ given by Lemma \ref{lem:flat-2d-high-cycle} and $B \le |G|^{O(1)}$ be the bound on $|\Lambda|$ composed from Lemmas \ref{lem:flat-2d-high-cycle} and \ref{lem:zigzag-free-cycle-lps}. 
    We put $D = \max\{D_0, B \cdot \norm{T}\} \le |G|^{O(1)}$ here.

    Fix $q\in Q$ and $\Vec{u}, \Vec{v} \in [D, \infty)^2$ with $q(\Vec{u}) \xrightarrow{*}_{\mathbb{N}^2} q(\Vec{v})$. 
    By Lemma \ref{lem:flat-2d-high-cycle} there is an LPS $\sigma$ compatible to $G$ such that $q(\Vec{u}) \xrightarrow{\sigma}_{\mathbb{N}^2} q(\Vec{v})$. Apply Lemma \ref{lem:zigzag-free-cycle-lps} to $\sigma$, we get a zigzag-free LPS $\sigma'$ compatible to $G$ with $q(\Vec{u}) \xrightarrow{\sigma'}_{\mathbb{Z}^2} q(\Vec{v})$ and $|\sigma'| \le B$.
    Lemma \ref{lem:zigzag-free-LPS-high-reach} and our choice of $D$ then show that indeed we have $q(\Vec{u}) \xrightarrow{\sigma'}_{\mathbb{N}^2} q(\Vec{v})$.
    Let $Z$ be an orthant containing $\Vec{v} - \Vec{u}$.
    We finally apply Lemma \ref{lem:zigzag-cycle-lps-orthant} to $\sigma'$ to obtain an LPS $\Lambda$ compatible to $G$ that is zigzag-free in $Z$ such that $q(\Vec{u}) \xrightarrow{\Lambda}_{\mathbb{N}^2} q(\Vec{v})$ and $|\Lambda| \le (B + \norm{T})^{O(1)} \le |G|^{O(1)}$.
\end{proof}

\subsection{Sign Reflecting Projection}
\label{sec:srp}

In this section we develop a technical tool called the ``sign reflecting projection'' which will be used to project a geometrically $2$-dimensional $d$-VASS to a $2$-VASS so that cyclic runs with effect belonging to a certain orthant can be safely projected back. We first introduce the notion of projection of vectors.

\paragraph*{Projection.} 
Let $I\subseteq [d]$ be a subset of indices. Given a vector $\Vec{u} \in \mathbb{X}^d$ where $\mathbb{X}$ can be one of $\mathbb{Q}, \mathbb{N}, \mathbb{Z}$, we define the \emph{projection of $\Vec{u}$ onto indices in $I$} as a function $\Vec{u}|_I \in \mathbb{X}^{I}$ given by 
\begin{equation}
    (\Vec{u}|_I)(i) = \Vec{u}(i) \text{ for all } i\in I.
\end{equation}
We tacitly identify the function $\Vec{u}|_I \in \mathbb{X}^I$ as a vector in $\mathbb{X}^{|I|}$.
Let $i\in[d]$, we abbreviate $\Vec{u}^{-i}$ for the vector $\Vec{u}|_{[d]\setminus \{i\}}$. The definition of projection is naturally extended to set of vectors $V$:%
\begin{equation}
    V|_I := \bigl\{ 
        \Vec{v}|_I : \Vec{v}\in V
    \bigr\}.
\end{equation}
It should be clear that the projection of a vector space onto indices in $I$ is again a vector space in $\mathbb{Q}^{|I|}$, and the projection of an orthant $Z_t$ onto $I$ is an orthant in $\mathbb{Q}^{|I|}$ defined by $t|_I$.

\paragraph*{Sign reflecting projection.}
\begin{definition}
    Let $P \subseteq \mathbb{Q}^d$ be a vector space and $Z$ be an orthant in $\mathbb{Q}^d$. A set of indices $I\subseteq [d]$ is called a \emph{sign-reflecting projection for $P$ with respect to $Z$} if for any $\Vec{v} \in P$, $\Vec{v}|_I \in Z|_I$ implies $\Vec{v} \in Z$.
\end{definition}

Sign-reflecting projection helps us project the vectors in a vector space to some of its components so that the pre-image of a certain orthant still belongs to one orthant. Moreover, we can show that such a projection is one-to-one when restricted to that orthant. 

\begin{lemma}
    \label{lem:spp-inject}
    Let $P \subseteq \mathbb{Q}^d$ be a vector space and $Z$ be an orthant in $\mathbb{Q}^d$. Let $I\subseteq [d]$ be a sign-reflecting projection for $P$ w.r.t.\ $Z$. Then every vector $\Vec{v} \in P\cap Z$ is uniquely determined by $\Vec{v}|_I$. In other words, for any $\Vec{v}, \Vec{v}' \in P\cap Z$, $\Vec{v}|_I = \Vec{v}'|_I$ implies $\Vec{v} = \Vec{v}'$.
\end{lemma}

\begin{proof}
    Suppose otherwise, there are vectors $\Vec{v} \ne \Vec{v}' \in P\cap Z$ with $\Vec{v}|_I = \Vec{v}'|_I$. Then there must be an index $j \in [d]$ with $\Vec{v}(j) \ne \Vec{v}'(j)$. Note that $\Vec{v} - \Vec{v}' \in P$ and $\Vec{v}' - \Vec{v} \in P$ cannot belong to the same orthant, as the signs of their $j$th component differ. On the other hand, $(\Vec{v} - \Vec{v}')|_I = (\Vec{v}' - \Vec{v})|_I = \Vec{0} \in Z|_I$. Since $I$ is a sign-reflecting projection, we must have $\Vec{v} - \Vec{v}' \in Z$ and $\Vec{v}' - \Vec{v} \in Z$, a contradiction.
\end{proof}

We mainly care about finding sign-reflecting projections for a plane $P$, that is, a $2$-dimensional subspace of $\mathbb{Q}^d$. The main technical theorem is stated as follows.

\begin{restatable}{theorem}{ThmTwoDProj}
    \label{thm:2d-proj}
    For any dimension $d \ge 2$, let $P \subseteq \mathbb{Q}^d$ be a plane (i.e.\ a $2$-dimensional subspace), and $Z$ be an orthant in $\mathbb{Q}^d$ such that $P\cap Z$ contains two linearly independent vectors. Then there is a sign-reflecting projection $I$ for $P$ w.r.t.\ $Z$ such that $|I| = 2$.
\end{restatable}

The proof of \ref{thm:2d-proj} makes use of the Farkas-Minkowski-Weyl's Theorem. We first introduce some additional definitions about \emph{cones}.

A \emph{polyhedral cone} is a set of the form $\{\Vec{v} \in \mathbb{Q}^d : A\Vec{v}\le 0\}$ where $A$ is a rational matrix. Let $V\subseteq \mathbb{Q}^d$ be a set of vectors, the \emph{cone generated by $V$} is the following set:
\begin{equation}
    \Cone(V) := \bigcup_{k\in\mathbb{N}}\{\lambda_1 x_1 + \cdots + \lambda_k x_k : \lambda_1, \ldots, \lambda_k\in \mathbb{Q}_{\ge0}, x_1, \ldots, x_k\in V\}.
\end{equation}
A \emph{finitely generated cone} is a cone generated by a finite set. The Farkas-Minkowski-Weyl's Theorem states that:

\begin{theorem}[Farkas-Minkowski-Weyl]
    \label{thm:FMW}
    A set $C\subseteq \mathbb{Q}^d$ is a polyhedral cone if and only if it is a finitely generated cone.
\end{theorem}

Let $P$ be a plane and $Z$ be an orthant in $\mathbb{Q}^d$, one easily verifies that $P\cap Z$ is a polyhedral cone, where the inequality $A\Vec{v} \le 0$ asserts that (i) $\Vec{v}$ is orthogonal to all normal vectors of $P$, and (ii) $\Vec{v}$ belongs to the orthant $Z$.
Theorem \ref{thm:FMW} then shows that the intersection of a plane and an orthant is a finitely generated cone. We need a stronger fact that it is generated by at most 2 vectors.

\begin{lemma}
    \label{lem:2-vecs}
    Let $P\subseteq\mathbb{Q}^d$ be a plane and $Z$ be an orthant in $\mathbb{Q}^d$. Then there exists $V\subseteq \mathbb{Q}^d$ such that $P\cap Z = \Cone(V)$ and $|V| \le 2$.
\end{lemma}

\begin{proof}
    Let $P$ be a plane and $Z$ be an orthant in $\mathbb{Q}^d$. The existence of a finite set $V\subseteq\mathbb{Q}^d$ such that $P\cap Z = \Cone(V)$ is guaranteed by Theorem \ref{thm:FMW}. Let $V$ have minimal size. Towards a contradiction we assume that $V$ contains 3 distinct non-zero vectors $\Vec{u}, \Vec{v}, \Vec{w}$. These vectors must be linearly dependent since $V\subseteq P$ and $P$ is 2-dimensional. Suppose $\Vec{w} = \lambda_1 \Vec{u} + \lambda_2 \Vec{v}$. We can make sure that $\lambda_1$ and $\lambda_2$ are of the same sign since otherwise we just rename $\Vec{u}, \Vec{v}, \Vec{w}$ properly. If $\lambda_1, \lambda_2 \ge 0$ then $\Cone(V) = \Cone(V\setminus\{\Vec{w}\})$, contradicting to the minimality of $|V|$. Thus $\lambda_1, \lambda_2 \le 0$. But then $\Vec{u}, \Vec{v}$ and $\Vec{w}$ cannot belong to the same orthant, contradicting to $V\subseteq Z$. Therefore, we must have $|V| \le 2$.
\end{proof}

We can now prove Theorem \ref{thm:2d-proj}.

\begin{proof}[Proof of Theorem \ref{thm:2d-proj}]
    Let $P \subseteq \mathbb{Q}^d$ be a plane, and $Z$ be an orthant in $\mathbb{Q}^d$ such that $P\cap Z$ contains two linearly independent vectors. Without loss of generality we assume that $Z = \mathbb{Q}_{\ge0}$ is the first orthant. By Lemma \ref{lem:2-vecs} there are non-zero vectors $\Vec{u}, \Vec{v}\in\mathbb{Q}^d$ such that $P\cap Z = \Cone(\{\Vec{u}, \Vec{v}\})$. We emphasize that $\Vec{u}, \Vec{v}\ge \Vec{0}$.

    Clearly $\Vec{u}, \Vec{v}$ are linearly independent since the cone generated by them contains two linearly independent vectors. We claim that $\Supp(\Vec{u})$ and $\Supp(\Vec{v})$ are incomparable under $\subseteq$. 
    Otherwise, suppose $\Supp(\Vec{v})\subseteq \Supp(\Vec{u})$, then $\lambda \Vec{u} - \Vec{v} \in P\cap Z$ for sufficiently large $\lambda > 0$. As $P\cap Z = \Cone(\{\Vec{u}, \Vec{v}\})$, we have $\lambda \Vec{u} - \Vec{v} = \alpha \Vec{u} + \beta \Vec{v}$ for some $\alpha, \beta\ge 0$. Then $(\beta + 1)\Vec{v} = (\lambda - \alpha) \Vec{u}$, contradicting to the linear independence of $\Vec{u}$ and $\Vec{v}$. 
    Symmetrically, we cannot have $\Supp(\Vec{u})\subseteq \Supp(\Vec{v})$. 
    Now take $i \in \Supp(\Vec{u}) \setminus \Supp(\Vec{v})$ and $j\in \Supp(\Vec{v})\setminus\Supp(\Vec{u})$. 
    We show that $I = \{i, j\}$ is a sign-reflecting projection for $P$ w.r.t.\ $Z$. 
    Indeed, let $\Vec{w}\in P = \Span(\{\Vec{u}, \Vec{v}\})$ be such that $\Vec{w}(i), \Vec{w}(j)\ge 0$. Suppose $\Vec{w} = \alpha \Vec{u} + \beta \Vec{v}$ for some $\alpha, \beta \in \mathbb{Q}$. Then $\Vec{w}(i) = \alpha \Vec{u}(i) \ge 0$ implies $\alpha \ge 0$, and $\Vec{w}(j) = \beta \Vec{v}(j) \ge 0$ implies $\beta \ge 0$. Thus $\Vec{w} \in \Cone(\{\Vec{u}, \Vec{v}\})\subseteq Z$. 
\end{proof}

Finally, we remark that the requirement in Theorem \ref{thm:2d-proj} that $P\cap Z$ contains two linearly independent vectors is not essential: 

\begin{lemma}
    \label{lem:2d-orthant}
    Let $P\subseteq \mathbb{Q}^d$ be a plane, and $\Vec{v}\in P$ be a non-zero vector. There exists an orthant $Z$ such that $\Vec{v}\in Z$ and $P\cap Z$ contains two linearly independent vectors.
\end{lemma}

\begin{proof}
    Let $\Vec{v}\in P$ be non-zero. Pick $\Vec{u}\in P$ such that $P = \Span(\{ \Vec{u}, \Vec{v} \})$. Choose $t > 0$ large enough such that $t|\Vec{v}(i)| > |\Vec{u}(i)|$ for all $i \in \Supp(\Vec{v})$. Observe that $\Vec{v}$ and $t\Vec{v} + \Vec{u}$ are linearly independent and lie in the same orthant $Z$ as we desired.
\end{proof}

\subsection{Reachability Far Away from Axes}
\label{sec:flat-high-runs}

In this section, we prove that runs far away the axes in $\mathbb{N}^d$ can be captured by small linear path schemes. We shall focus on cycles at first:

\begin{restatable}{lemma}{flatHighCycles}
    \label{lem:flat-high-cycles}
    For every $d$-VASS $G = (Q, T)$ that is geometrically $2$-dimensional, there exists a natural numbers $D \le |G|^{O(1)}$ such that for every state $q \in Q$ and $\Vec{u}, \Vec{v} \in [D, \infty)^d$ with $q(\Vec{u})\xrightarrow{*} q(\Vec{v})$, there exists an LPS $\Lambda$ compatible to $G$ such that $q(\Vec{u}) \xrightarrow{\Lambda} q(\Vec{v})$ and $|\Lambda| \le |G|^{O(1)}$.
\end{restatable}

\begin{proof}
    Let $D_0$ be the constants $D$ as in Corollary \ref{cor:zigzag-free-in-Z-cycle-lps} and $B \le |G|^{O(1)}$ be the upper bound of lengths of LPSes given by Corollary \ref{cor:zigzag-free-in-Z-cycle-lps}. We put $D = \max\{ D_0, B\cdot \norm{T}\} \le |G|^{O(1)}$ here. 
    
    Let $q\in Q$ and $\Vec{u}, \Vec{v} \in [D, \infty)^d$ be such that $q(\Vec{u}) \xrightarrow{*} q(\Vec{v})$. Clearly $\Vec{v} - \Vec{u} \in \CycleSpace(G)$. Let $Z$ be the orthant given by Lemma \ref{lem:2d-orthant} such that $Z$ contains $\Vec{v} - \Vec{u}$ and that $\CycleSpace(G)\cap Z$ contains two linearly independent vectors. Now Theorem \ref{thm:2d-proj} shows that there exists a sign-reflecting projection $I\subseteq [d]$ for $\CycleSpace(G)$ w.r.t.\ $Z$, where $|I| = 2$.

    We define the $2$-VASS $G|_I = (Q, T|_I)$ as the projection of $G$ onto indices in $I$, where the transition set $T|_I$ is given by
    \begin{equation}
        T|_I = \{(p, \Vec{a}|_I, q) : (p, \Vec{a}, q) \in T\}.
    \end{equation}
    For every transition $t' = (p, \Vec{a}', q)\in T|_I$ we fix $g(t')$ to be an arbitrary transition $t = (p, \Vec{a}, q)\in T$ such that $\Vec{a}|_I = \Vec{a}'$. The mapping $g$ is extended to paths and cycles naturally.

    Note that $\Vec{u}|_I, \Vec{v}|_I \in [D, \infty)^2$ and $q(\Vec{u}|_I) \xrightarrow{*} q(\Vec{v}|_I)$ in $G|_I$. By Corollary \ref{cor:zigzag-free-in-Z-cycle-lps} there is an LPS $\sigma = \alpha_0\beta_1^*\alpha_1 \ldots  \beta_k^*\alpha_k$ compatible to $G|_I$ that is zigzag-free in $Z|_I$ such that $q(\Vec{u}|_I) \xrightarrow{\sigma} q(\Vec{v}|_I)$. Moreover, $|\sigma|\le B \le D$. Consider the LPS $\Lambda$ given by 
    \begin{equation}
        \Lambda = g(\alpha_0)g(\beta_1)^*g(\alpha_1) \ldots g(\beta_k)^*g(\alpha_k).
    \end{equation}
    Clearly $\Lambda$ is an LPS compatible to $G$ from state $q$ to $q$. Note that for all $i \in [k]$ we have $\Delta(g(\beta_i)) \in \CycleSpace(G)$ and $\Delta(g(\beta_i))|_I =\Delta(\beta_i) \in Z|_I$, thus $\Delta(g(\beta_i)) \in Z$ by the fact that $I$ is sign-reflecting. We have shown that $\Lambda$ is zigzag-free in $Z$. Next we prove that $q(\Vec{u})\xrightarrow{\Lambda} q(\Vec{v})$.

    As $q(\Vec{u}|_I) \xrightarrow{\sigma} q(\Vec{v}|_I)$, there must be numbers $e_1, \ldots, e_k\in\mathbb{N}$ such that 
    \begin{equation}
        \label{eq:reach-in-projected-2vass}
        q(\Vec{u}|_I)\xrightarrow{\alpha_0\beta_1^{e_1}\alpha_1\ldots \beta_k^{e_k}\alpha_k} q(\Vec{v}|_I).
    \end{equation}
    Let $\rho$ be the path in $G$ given by
    \begin{equation}
        \rho := g(\alpha_0) g(\beta_1)^{e_1} g(\alpha_1) \ldots g(\beta_k)^{e_k} g(\alpha_k),
    \end{equation}
    We shall show that $q(\Vec{u}) \xrightarrow{\rho}_{\mathbb{Z}^d} q(\Vec{v})$ holds in $G$. Then by Lemma \ref{lem:zigzag-free-LPS-high-reach} and our choice of $D$ we have indeed $q(\Vec{u}) \xrightarrow{\rho}_{\mathbb{N}^d} q(\Vec{v})$, witnessing the fact that $q(\Vec{u})\xrightarrow{\Lambda} q(\Vec{v})$.
    Now we only need to check that $\Delta(\rho) = \Vec{v} - \Vec{u}$. Recall that $\rho$ is a cycle, so $\Delta(\rho) \in \CycleSpace(G)$. By (\ref{eq:reach-in-projected-2vass}) and the definition of $g$ we have $\Delta(\rho)|_I = \Vec{v}|_I - \Vec{u}|_I \in Z|_I$. We then have $\Delta(\rho) \in Z$ since $I$ is sign-reflecting. By Lemma \ref{lem:spp-inject}, any vector in $\CycleSpace(G)\cap Z$ is uniquely determined by its components in $I$. As $\Delta(\rho)|_I = (\Vec{v} - \Vec{u})|_I$ we must have $\Delta(\rho) = \Vec{v} - \Vec{u}$ as claimed.

    Finally, we bound the length of $\Lambda$. Just note that $|\Lambda| = |\sigma| \le |G|^{O(1)}$.
\end{proof}

Runs that stay far away from the axes can be decomposed into a short path interleaved by at most $|Q|$ cycles, where each cycle can be handled by Lemma \ref{lem:flat-high-cycles}. This leads to the second result we will prove in this section:
\begin{restatable}{lemma}{flatHighRuns}
    \label{lem:flat-high-runs}
    For every $d$-VASS $G = (Q, T)$ that is geometrically $2$-dimensional, there exists a natural number $D \le |G|^{O(1)}$ such that for every pair of configurations $p(\Vec{u}), q(\Vec{v}) \in Q\times [D, \infty)^d$ with $p(\Vec{u})\xrightarrow{*}_{[D, \infty)^d} q(\Vec{v})$, there exists an LPS $\Lambda$ compatible to $G$ such that $p(\Vec{u}) \xrightarrow{\Lambda} q(\Vec{v})$ and $|\Lambda| \le |G|^{O(1)}$.
\end{restatable}

\begin{proof}
    We put $D$ to be the same constant as in Lemma \ref{lem:flat-high-cycles}. Denote $\mathbb{O} := [D, \infty)^d$. Let $p(\Vec{u}), q(\Vec{v}) \in Q \times \mathbb{O}$ be such that $p(\Vec{u}) \xrightarrow{*}_{\mathbb{O}} q(\Vec{v})$.
    
    Let $\pi$ be a run witnessing $p(\Vec{u}) \xrightarrow{*}_{\mathbb{O}} q(\Vec{v})$. By extracting maximal cycles from $\pi$, we can write $\pi$ in the following form:
    \begin{equation}
        \pi : p(\Vec{u}) 
        \xrightarrow{\alpha_0}_{\mathbb{O}} q_1(\Vec{u}_1) 
        \xrightarrow{\beta_1}_{\mathbb{O}} q_1(\Vec{v}_1)
        \xrightarrow{\alpha_1}_{\mathbb{O}} q_2(\Vec{u}_2) 
        \cdots 
        q_k(\Vec{u}_k) 
        \xrightarrow{\beta_k}_{\mathbb{O}} q_k(\Vec{v}_k)
        \xrightarrow{\alpha_k}_{\mathbb{O}} q(\Vec{v})
    \end{equation}
    where 
    \begin{itemize}
        \item $q_1, \ldots, q_k \in Q$ are pairwise distinct, so $k \le |Q|$;
        \item $\Vec{u}_1, \ldots, \Vec{u}_k, \Vec{v}_1, \ldots, \Vec{v}_k \in \mathbb{O}$;
        \item $\alpha_0\alpha_1\ldots \alpha_k$ is a simple path, so $|\alpha_0| + \cdots + |\alpha_k| \le |Q|$.
    \end{itemize}

    By Lemma \ref{lem:flat-high-cycles}, for each cycle $\beta_i$ there exists an LPS $\Lambda_i$ compatible to $G$ such that $q_i(\Vec{u}_i) \xrightarrow{\Lambda_i} q_i(\Vec{v}_i)$ and $|\Lambda_i| \le |G|^{O(1)}$. Let 
    \begin{equation}
        \Lambda = \alpha_0\Lambda_1\alpha_1\ldots \alpha_{k-1}\Lambda_k \alpha_k,
    \end{equation}
    it is easy to see that $\Lambda$ is compatible to $G$, and $p(\Vec{u}) \xrightarrow{\Lambda} q(\Vec{v})$ holds. 
    Finally, we bound the length of $\Lambda$:
    \begin{equation}
        |\Lambda| = \sum_{i = 0}^k |\alpha_i| + \sum_{i = 1}^k |\Lambda_i| \le |Q| + |Q| \cdot |G|^{O(1)} \le |G|^{O(1)}.
    \end{equation}
\end{proof}

We remark that Lemma \ref{lem:flat-high-runs} captures only the runs that stay in the region $[D, \infty)^d$, while the run derived from the LPS $\Lambda$ given by it may go out of $[D, \infty)^d$ but stays in $\mathbb{N}^d$.

\subsection{Reachability Near the Axes}
\label{sec:flat-low-runs}

We have captured runs that stay far away from axes by small LPSes in Section \ref{sec:flat-high-runs}. In this section we study the complementary region. Let $D \in \mathbb{N}$ be any natural number, we define the region 
\begin{equation}
    \mathbb{L}[D] := \{\Vec{u} \in \mathbb{N}^d : \Vec{u}(i) \le D \text{ for some }i \in [d]\}.
\end{equation}
The following is to be proved:
\begin{lemma}
    \label{lem:flat-near-axes}
    Let $D\in\mathbb{N}$ and let $G = (Q, T)$ be a $d$-VASS that is geometrically $2$-dimensional. For every pair of configurations $p(\Vec{u}), q(\Vec{v}) \in Q\times \mathbb{L}[D]$ with $p(\Vec{u}) \xrightarrow{*}_{\mathbb{L}[D]} q(\Vec{v})$ there exists an LPS $\Lambda$ compatible to $G$ such that $p(\Vec{u}) \xrightarrow{\Lambda} q(\Vec{v})$, and $|\Lambda| \le (|G| + D)^{O(1)}$.
\end{lemma}

Let's first subdivide the region $\mathbb{L}[D]$ into simpler regions. For any $D \in \mathbb{N}$ and $i \in [d]$, we define 
\begin{equation}
    \mathbb{L}_i[D] := \{\Vec{u} \in \mathbb{N}^d : \Vec{u}(i) \le D\}.
\end{equation}
Clearly $\mathbb{L}[D] = \bigcup_{i \in [d]} \mathbb{L}_i[D]$. We shall prove a stronger result that runs in a union of any number of $\mathbb{L}_i[D]$'s can be captured by LPSes. For $\mathcal{I} \subseteq [d]$ we write $\mathbb{L}_\mathcal{I}[D] := \bigcup_{i \in \mathcal{I}}\mathbb{L}_i[D]$.

\begin{lemma}
    \label{lem:flat-near-axes-subset}
    Let $D\in\mathbb{N}$, $\mathcal{I} \subseteq [d]$, and let $G = (Q, T)$ be a $d$-VASS that is geometrically $2$-dimensional. For every pair of configurations $p(\Vec{u}), q(\Vec{v}) \in Q\times \mathbb{L}_\mathcal{I}[D]$ with $p(\Vec{u}) \xrightarrow{*}_{\mathbb{L}_\mathcal{I}[D]} q(\Vec{v})$ there exists an LPS $\Lambda$ compatible to $G$ such that $p(\Vec{u}) \xrightarrow{\Lambda} q(\Vec{v})$, and $|\Lambda| \le (|G| + D)^{O(1)}$.
\end{lemma}

Note that Lemma \ref{lem:flat-near-axes} follows directly from Lemma \ref{lem:flat-near-axes-subset} by taking $\mathcal{I} = [d]$. In the rest we will prove Lemma \ref{lem:flat-near-axes-subset} by induction on the size of $\mathcal{I}$.

\subsubsection{The Base Case: Reachability with One Bounded Component}

A run that stays in $\mathbb{L}_i[D]$ has its $i$th component bounded by $D$. In this case we can simply encode this component in the states and reduce the dimension by one. Recall that we are proving Theorem \ref{thm:geo-2vass-flat-lps} by induction on the dimension $d$, so here we assume that Theorem \ref{thm:geo-2vass-flat-lps} already holds for $(d-1)$-VASSes as the induction hypothesis.

\begin{lemma}
    \label{lem:bounded-counter-reduce-dim}
    Let $0 \le B \le C$ and $i \in [d]$, let $\mathbb{L}_i[B, C]$ denote the region $\{\Vec{u} \in \mathbb{N}^d : B \le \Vec{u}(i) \le C\}$. Let $G = (Q, T)$ be a $d$-VASS that is geometrically $2$-dimensional. For every pair of configurations $p(\Vec{u}), q(\Vec{v}) \in Q\times \mathbb{L}_i[B, C]$ with $p(\Vec{u}) \xrightarrow{*}_{\mathbb{L}_i[B, C]} q(\Vec{v})$ there exists an LPS $\Lambda$ compatible to $G$ such that $p(\Vec{u}) \xrightarrow{\Lambda} q(\Vec{v})$, and $|\Lambda| \le (|G| + (C - B))^{O(1)}$.
\end{lemma}

\begin{proof}
    Let $G^{-i} = (Q^{-i}, T^{-i})$ be the $(d-1)$-VASS obtained from $G$ where 
    \begin{align}
        Q^{-i} &= \{(q, z) : q\in Q, z\in [B, C]\},\\
        T^{-i} &= \{ ( (p, z), \Vec{a}^{-i}, (q, z') ) : (p, \Vec{a}, q) \in T, z' - z = \Vec{a}(i) \}.
    \end{align}
    For each transition $t = ( (p, z), \Vec{a}^{-i}, (q, z') ) \in T^{-i}$ let $g(t)$ be the unique transition $(p, \Vec{a}, q) \in T$ where $\Vec{a}(i) = z' - z$. The mapping $g$ is extended to linear path schemes compatible to $G^{-i}$ as a word morphism.

    Note that $G^{-i}$ is still geometrically $2$-dimensional as $G$ is. Moreover, $|Q^{-i}| = |Q| \cdot (C - B + 1)$, and $\norm{T^{-i}} \le \norm{T}$, so $|G^{-i}| \le (|G| + C - B)^{O(1)}$. Let $p(\Vec{u}), q(\Vec{v}) \in Q\times \mathbb{L}_i[B, C]$ be configurations in $G$ such that $p(\Vec{u}) \xrightarrow{*}_{\mathbb{L}_i[B, C]} q(\Vec{v})$ holds in $G$. Then $(p, \Vec{u}(i))(\Vec{u}^{-i}) \xrightarrow{*}_{\mathbb{N}^{d-1}} (q, \Vec{v}(i))(\Vec{v}^{-i})$ holds in $G^{-i}$. Recall we have assumed that Theorem \ref{thm:geo-2vass-flat-lps} holds for dimension $d - 1$. In particular, it holds for $G^{-i}$. Thus there exists an LPS $\Lambda$ compatible to $G^{-i}$ such that $(p, \Vec{u}(i))(\Vec{u}^{-i}) \xrightarrow{\Lambda}_{\mathbb{N}^{d-1}} (q, \Vec{v}(i))(\Vec{v}^{-i})$. Note that $g(\Lambda)$ is an LPS compatible to $G$. Since the $i$th coordinate is encoded in the state of $G^{-i}$, one easily checks that $p(\Vec{u}) \xrightarrow{g(\Lambda)}_{\mathbb{N}^d} q(\Vec{v})$ also holds in $G$.

    As for the size of $\Lambda$, just note that $|g(\Lambda)| = |\Lambda| \le |G^{-i}|^{O(1)} \le (|G| + C - B)^{O(1)}$.
\end{proof}

The base case of Lemma \ref{lem:flat-near-axes-subset} where $\mathcal{I}$ contains only one index is just an easy corollary of Lemma \ref{lem:bounded-counter-reduce-dim} by taking $B = 0$ and $C = D$. The general form of Lemma \ref{lem:bounded-counter-reduce-dim} will be used later.

\subsubsection{The Inductive Proof of Lemma \ref{lem:flat-near-axes-subset}}

We now consider $\mathbb{L}_\mathcal{I}[D]$-reachability for $\mathcal{I}\subseteq [d]$ with $|\mathcal{I}| > 1$, assuming that Lemma \ref{lem:flat-near-axes-subset} holds for all subsets $\mathcal{I}' \subsetneq \mathcal{I}$. First we identify a special case where $\CycleSpace(G)$ is orthogonal to one of the axes, which makes $G$ essentially a $(d-1)$-VASS.

\begin{lemma}
    \label{lem:flat-cyc-degenerate}
    Let $G = (Q, T)$ be a $d$-VASS that is geometrically $2$-dimensional. Suppose there exists $i\in [d]$ such that $i \notin \Supp(\CycleSpace(G))$. Then for every pair of configurations $p(\Vec{u}), q(\Vec{v}) \in Q\times \mathbb{N}^d$ with $p(\Vec{u}) \xrightarrow{*} q(\Vec{v})$ there exists an LPS $\Lambda$ compatible to $G$ such that $p(\Vec{u}) \xrightarrow{\Lambda} q(\Vec{v})$, and $|\Lambda| \le |G|^{O(1)}$.
\end{lemma}

\begin{proof}
    Let $p(\Vec{u}), q(\Vec{v}) \in Q\times \mathbb{N}^d$ be two configurations of $G$ with $p(\Vec{u}) \xrightarrow{*} q(\Vec{v})$. We show that there exists an LPS $\Lambda$ compatible to $G$ such that $p(\Vec{u}) \xrightarrow{\Lambda} q(\Vec{v})$ and $|\Lambda| \le |G|^{O(1)}$.

    First we claim that for any configurations $r(\Vec{w}) \in Q\times\mathbb{N}^d$ reachable from $p(\Vec{u})$, we have $| \Vec{w}(i) - \Vec{u}(i) | \le |Q|\cdot \norm{T}$. Otherwise, assume there is a run $\pi: p(\Vec{u}) \xrightarrow{*} r(\Vec{w})$ such that $| \Vec{w}(i) - \Vec{u}(i) | > |Q|\cdot \norm{T}$. Assume w.l.o.g.\ it is the case that $\Vec{w}(i) - \Vec{u}(i) > |Q|\cdot \norm{T}$. Write $\pi$ in the following form as a sequence of sub-runs that increase the $i$th component like a stair:
    \begin{equation}
        \pi: p(\Vec{u}) = q_0(\Vec{u}_0) \xrightarrow{*} q_1(\Vec{u}_1) \xrightarrow{*} \cdots \xrightarrow{*} q_k(\Vec{u}_k) \xrightarrow{*} r(\Vec{w})
    \end{equation}
    such that $\Vec{u}_k(i) \ge \Vec{w}(i) > \Vec{u}(i) + |Q|\cdot \norm{T}$ and for all $j \in [k]$, $\Vec{u}_{j}(i) > \Vec{u}_{j-1}(i)$ and every configuration on the sub-run from $q_{j-1}(\Vec{u}_{j-1})$ to $q_{j}(\Vec{u}_{j})$ has its $i$-th component no greater than $\Vec{u}_{j-1}(i)$. Clearly $\Vec{u}_{j}(i) - \Vec{u}_{j-1}(i) \le \norm{T}$, so $k \ge (\Vec{u}_k(i) - \Vec{u}(i)) / \norm{T} > |Q|$. By the Pigeonhole Principle there exists $j_1 < j_2$ such that $q_{j_1} = q_{j_2}$. The sub-run from $q_{j_1}(\Vec{u}_{j_1})$ to $q_{j_2}(\Vec{u}_{j_2})$ is a cycle which has positive effect in its $i$-th component, contradicting to the assumption that $i\notin\Supp(\CycleSpace(G))$.

    We have shown that $p(\Vec{u}) \xrightarrow{*}_{\mathbb{L}_i[B, C]} q(\Vec{v})$ where $B = \Vec{u}(i) - |Q|\cdot \norm{T}$ and $C = \Vec{u}(i) + |Q|\cdot \norm{T}$. Note that $C - B = 2\cdot |Q| \cdot \norm{T} \le |G|^{O(1)}$. The result follows directly from Lemma \ref{lem:bounded-counter-reduce-dim}.
\end{proof}

Now we prove Lemma \ref{lem:flat-near-axes-subset}.

\begin{proof}[Proof of Lemma \ref{lem:flat-near-axes-subset}]
    We prove by induction on $|\mathcal{I}|$. The base case where $|\mathcal{I}| = 1$ follows from Lemma \ref{lem:bounded-counter-reduce-dim} by taking $B = 0$ and $C = D$. Now assume that $\mathcal{I} = \{i\} \uplus \mathcal{I}'$ where $i \in \mathcal{I}$ is arbitrary and $\mathcal{I}' = \mathcal{I} \setminus\{i\}$ is non-empty. Fix $p(\Vec{u}), q(\Vec{v}) \in Q\times \mathbb{N}^d$ with $p(\Vec{u}) \xrightarrow{*}_{\mathbb{L}_{\mathcal{I}}[D]} q(\Vec{v})$, we will exhibit an LPS $\Lambda$ compatible to $G$ such that $p(\Vec{u}) \xrightarrow{\Lambda} q(\Vec{v})$.

    Fix a run $\pi$ witnessing $p(\Vec{u}) \xrightarrow{*}_{\mathbb{L}_{\mathcal{I}}[D]} q(\Vec{v})$ such that no configuration occurs more than once on $\pi$. 
    We assume that $\pi$ must visit both $\mathbb{L}_i[D]$ and $\mathbb{L}_{\mathcal{I}'}[D]$, otherwise we can directly apply the induction hypothesis.
    Observe that $\pi$ can be written in the following form as an alternating sequence of sub-runs lie in either $\mathbb{L}_i[D]$ or $\mathbb{L}_{\mathcal{I}'}[D]$:
    \begin{equation}
        \pi : p(\Vec{u}) \xrightarrow{\sigma}_{\mathbb{R}_0} 
            q_0(\Vec{u}_0) \xrightarrow{\pi_1}_{\mathbb{R}_1} 
            q_1(\Vec{u}_1) \xrightarrow{\pi_2}_{\mathbb{R}_2} \cdots 
            \xrightarrow{\pi_{n}}_{\mathbb{R}_{n}} 
            q_{n}(\Vec{u}_{n}) \xrightarrow{\rho}_{\mathbb{R}_{n+1}} 
            q(\Vec{u})
    \end{equation}
    where 
    \begin{itemize}
        \item $\mathbb{R}_h \in \{\mathbb{L}_{i}[D], \mathbb{L}_{\mathcal{I}'}[D]\}$ for $h = 0, \ldots, n+1$, and $\mathbb{R}_h \ne \mathbb{R}_{h - 1}$ for $h\in [n+1]$,
        \item $\Vec{u}_h \in \mathbb{J}_{i\ell}$ for some (not necessarily unique) $\ell\in \mathcal{I}'$ for each $h = 0, \ldots, n$, where 
        \begin{equation}
            \mathbb{J}_{i\ell} := \{\Vec{w} \in \mathbb{N}^d : \Vec{w}(i), \Vec{w}(\ell) \le D + \norm{T}\}.
        \end{equation}
        (Note that we increase the bound by $\norm{T}$ so that any one-step transition from $\mathbb{L}_i[D]$ to $\mathbb{L}_\ell[D]$ lies completely in $\mathbb{J}_{i\ell}$.)
    \end{itemize}
    
    Note that the prefix $\sigma$ and the suffix $\rho$ can be captured by LPSes of length $(|G| + D)^{O(1)}$ by induction hypothesis. In the following we show that the infix $\overline{\pi} := \pi_1\pi_2\ldots \pi_n$ can also be captured by an LPS $\overline{\Lambda}$ of length $(|G| + D)^{O(1)}$. Then the entire run $\pi$ is captured by a concatenation of these three LPSes.

    Let $\pi_h$ be any sub-run of $\overline{\pi}$ such that $\mathbb{R}_h = \mathbb{L}_i[D]$ where $h \in [n]$. Expand $\pi_h$ as 
    \begin{equation}
        \pi_h : q_{h-1}(\Vec{u}_{h-1}) =: s_0(\Vec{v}_0) \xrightarrow{t_1}_{\mathbb{L}_i[D]} s_1(\Vec{v}_1) \xrightarrow{t_2}_{\mathbb{L}_i[D]} \cdots \xrightarrow{t_m}_{\mathbb{L}_i[D]} s_m(\Vec{v}_m) := q_{h}(\Vec{u}_{h}),
    \end{equation}
    where $t_1, \ldots, t_m \in T$ are transitions.
    Denote $D' := (D + \norm{T}) + |Q|\cdot D\cdot \norm{T}$,
    we classify $\pi_h$ into the following types:
    \begin{description}
        \item[Type 1:] \emph{There exists $\ell \in \mathcal{I}'$ such that $\Vec{v}_k(\ell) \le D'$ for all $k = 0, \ldots, m$.} 
        
        In this case, $\pi_h$ is actually an $\mathbb{L}_{\mathcal{I}'}[D']$-run.

        \item[Type 2($\ell$), where $\ell \in \mathcal{I}'$:] \emph{There exists $k\in [m]$ such that $\Vec{v}_k(\ell) > D'$, and $\Vec{u}_{h-1} = \Vec{v}_0 \in \mathbb{J}_{i\ell}$.}
        
        In this case $\Vec{v}_k(\ell) - \Vec{v}_0(\ell) > D' - (D + \norm{T}) = |Q|\cdot D\cdot \norm{T}$.
        Let $k$ be minimal such that $\Vec{v}_k(\ell) > D'$.
        For each $\kappa < k$ we define the index $\Next(\kappa)$ as the minimal index $\kappa'$ with $\kappa < \kappa' \le k$ and $\Vec{v}_{\kappa'}(\ell) > \Vec{v}_{\kappa}(\ell)$.
        Let $\kappa_1, \kappa_2, \ldots, \kappa_M$ be the sequence of indices such that $\kappa_1 = 0$ and $\kappa_{j + 1} = \Next(\kappa_j)$ for all $j = 1, \ldots, M - 1$, and $\kappa_M = k$. We must have $M > |Q| \cdot D$.
        Since $\Vec{v}_{\kappa_j}(i) \le D$ for all $\kappa_j$ (recall $\pi_h$ stays in $\mathbb{L}_i[D]$), by the Pigeonhole Principle there exist $1\le \kappa_{a} < \kappa_{b} \le k$ such that 
        \begin{equation}
            s_{\kappa_a} = s_{\kappa_{b}}, \quad 
            \Vec{v}_{\kappa_a}(i) = \Vec{v}_{\kappa_{b}}(i), \quad 
            \Vec{v}_{\kappa_a}(\ell) < \Vec{v}_{\kappa_{b}}(\ell).
        \end{equation}
        Therefore, the existence of a Type 2($\ell$) sub-run $\pi_h$ implies that $\CycleSpace(G)$ contains a vector $\Vec{r}$ such that $i\notin\Supp(\Vec{r})$ and $\ell \in \Supp(\Vec{r})$.
    \end{description}

    Recall that the run $\pi_h$ must start from a place in $\mathbb{J}_{i\ell}$ for some $\ell \in \mathcal{I}'$. So $\pi_h$ is either of Type 1 or of Type 2($\ell$) (or both, but we let Type 1 take precedence).

    If all the $\mathbb{L}_i[D]$-sub-runs $\pi_h$ are of Type 1, then $\overline{\pi}$ is indeed an $\mathbb{L}_{\mathcal{I}'}[D']$-run, and we are done by applying the induction hypothesis with $D' \le (|G| + D)^{O(1)}$ in place of $D$.
    Next we assume there exist Type 2($\ell$) sub-runs. By considering $\mathbb{L}_{\mathcal{I}'}[D']$ instead of $\mathbb{L}_{\mathcal{I}'}[D]$, we can rewrite $\overline{\pi}$ in the following form to get rid of all Type 1 sub-runs in $\mathbb{L}_i[D]$:
    \begin{align}
        \overline{\pi} : q_0(\Vec{u}_0) = 
        q_0'(\Vec{u}_0') \xrightarrow{\pi_1'}_{\mathbb{R}_1'} 
        q_1'(\Vec{u}_1') \xrightarrow{\pi_2'}_{\mathbb{R}_2'} \cdots 
        \xrightarrow{\pi_{N}'}_{\mathbb{R}_{N}'} 
        q_{N}'(\Vec{u}_{N}') 
        = q_n(\Vec{u}_n)
    \end{align}
    where 
    \begin{itemize}
        \item $\mathbb{R}_h' \in \{\mathbb{L}_{i}[D], \mathbb{L}_{\mathcal{I}'}[D']\}$ for $h = 0, \ldots, N+1$, and $\mathbb{R}_h' \ne \mathbb{R}_{h - 1}'$ for $h\in [N+1]$,
        \item $\Vec{u}_h' \in \mathbb{J}_{i\ell}$ for some (not necessarily unique) $\ell\in \mathcal{I}'$ for each $h = 0, \ldots, N$,
        \item for each $h \in [N]$ such that $\mathbb{R}_h' = \mathbb{L}_{i}[D]$, $\pi_h$ is of Type 2($\ell$) for some $\ell \in \mathcal{I}'$.
    \end{itemize}
    Let 
    \begin{equation}
        \mathcal{H}_\ell = \{h\in [N] : \mathbb{R}_h' = \mathbb{L}_i[D] \text{ and $\pi_h'$ is of Type 2($\ell$)}\}
    \end{equation}
    be the set of indices of those $\mathbb{L}_i[D]$-sub-runs that are of Type 2($\ell$). There remain two cases:

    \begin{description}
        \item[Case 1.] \emph{For all $\ell \in \mathcal{I}'$, we have $|\mathcal{H}_\ell| \le |Q| \cdot (D + \norm{T})^2$.}
        
        In this case, the total number of indices $h\in [N]$ such that $\mathbb{R}_{h} = \mathbb{L}_i[D]$ is at most $|\mathcal{I}'|\cdot |Q| \cdot (D + \norm{T})^2 \le d\cdot |Q| \cdot (D + \norm{T})^2$ and thus $N \le 2\cdot d \cdot |Q|(D + \norm{T})^2 + 1 \le (|G| + D)^{O(1)}$. Since by induction hypothesis each sub-run $\pi_h$ where $h \in [N]$ can be captured by an LPS $\Lambda_h$ of size $(|G| + D)^{O(1)}$, we are done by taking $\overline{\Lambda}$ as the concatenation of all $\Lambda_h$'s.

        \item[Case 2.] \emph{There exists $\ell \in \mathcal{I}'$ such that $|\mathcal{H}_\ell| > |Q| \cdot (D + \norm{T})^2$.}
        
        In this case, consider the following set of configurations:
        \begin{equation}
            \mathcal{C}_\ell := \{ q'_{h-1}(\Vec{u}'_{h-1}) : h\in \mathcal{H}_\ell \}
        \end{equation}
        which are the starting configurations of the Type 2($\ell$) sub-runs. By definition, we have $s(\Vec{v}) \in \mathbb{J}_{i\ell}$ for all $s(\Vec{v}) \in \mathcal{C}_\ell$. Since $|\mathcal{C}_\ell| = |\mathcal{H}_\ell| > |Q| \cdot (D + \norm{T})^2$, by the Pigeonhole Principle there exists two configurations $s_1(\Vec{v}_1), s_2(\Vec{v}_2) \in \mathcal{C}_\ell$ such that 
        \begin{equation}
            s_1 = s_2, \quad \Vec{v}_1(i) = \Vec{v}_2(i), \quad \Vec{v}_1(\ell) = \Vec{v}_2(\ell), \quad \Vec{v}_1 \ne \Vec{v}_2, \quad \text{and} \quad 
            s_1(\Vec{v}_1) \xrightarrow{*} s_2(\Vec{v}_2).
        \end{equation}
        This implies that $\Vec{r} = \Vec{s}_2 - \Vec{s}_1 \in \CycleSpace(G)$ is non-zero, and $i, \ell \notin \Supp(\Vec{r})$. On the other hand, the existence of Type 2($\ell$) sub-runs also implies that there exists $\Vec{r}'\in\CycleSpace(G)$ such that $i\notin \Supp(\Vec{r}')$ and $\ell \in \Supp(\Vec{r}')$. Clearly $\Vec{r}$ and $\Vec{r}'$ are linearly independent. As $G$ is geometrically $2$-dimensional, we have $\CycleSpace(G) = \Span\{\Vec{r}, \Vec{r}'\}$ and thus $i\notin \Supp(\CycleSpace(G))$. This situation is already solved by Lemma \ref{lem:flat-cyc-degenerate}.
        \qedhere
    \end{description}
\end{proof}

\subsection{Putting All Together: Proof of Theorem \ref{thm:geo-2vass-flat-lps}}

Having established Lemmas \ref{lem:flat-high-cycles}, \ref{lem:flat-high-runs} and \ref{lem:flat-near-axes}, we are now at a position to prove Theorem \ref{thm:geo-2vass-flat-lps}.

\flatGeoTwoVASS*

\begin{proof}[Proof of Theorem \ref{thm:geo-2vass-flat-lps}]
    Let $D$ be as in Lemmas \ref{lem:flat-high-cycles} and \ref{lem:flat-high-runs}. We define the following regions:
    \begin{equation}
        \mathbb{L} := \mathbb{L}[D + \norm{T}], \qquad
        \mathbb{O} := [D, \infty)^d
    \end{equation}
    where $\mathbb{L}[D + \norm{T}]$ is defined in Section \ref{sec:flat-low-runs}.

    Let $p(\Vec{u}), q(\Vec{v})$ be two configurations of $G$ such that $p(\Vec{u}) \xrightarrow{*} q(\Vec{v})$. We show that there exists an LPS $\Lambda$ compatible to $G$ such that $p(\Vec{u}) \xrightarrow{\Lambda} q(\Vec{v})$ and $|\Lambda| \le |G|^{O(1)}$.

    Let $\pi$ be a run from $p(\Vec{u})$ to $q(\Vec{v})$. We expand $\pi$ as the following form:
    \begin{equation}
        \pi: p(\Vec{u}) = q_0(\Vec{u}_0) \xrightarrow{t_1} q_1(\Vec{u}_1) \xrightarrow{t_2} \ldots \xrightarrow{t_n} q_n(\Vec{u}_n) = q(\Vec{v})
    \end{equation}
    where $t_1, \ldots, t_n \in T$. Define $K = \{k \in \{0, \ldots, n\} : \Vec{u}_k \in \mathbb{L} \cap \mathbb{O}\}$. For $k \in K$ let $\Next(k)$ be the minimal index $k' \in K$ such that $k' > k$ (if $k$ is the maximum element in $K$ then we set $\Next(k) = k$), and let $\ell(k)$ be the maximal index $k' \in K$ such that $q_k = q_{k'}$. We rewrite $\pi$ into the following form:
    \begin{align}
    \begin{split}
        \pi: 
        q_0(\Vec{u}_0) 
        &   \xrightarrow{\alpha_0}_{\mathbb{R}_0} q_{k_1}(\Vec{u}_{k_1}) 
            \xrightarrow{\gamma_1} q_{\ell(k_1)}(\Vec{u}_{\ell(k_1)})\\
        &   \xrightarrow{\alpha_1}_{\mathbb{R}_1} q_{k_2}(\Vec{u}_{k_2}) 
            \xrightarrow{\gamma_2} q_{\ell(k_2)}(\Vec{u}_{\ell(k_2)})\\
        &   \cdots \\
        &   \xrightarrow{\alpha_{h-1}}_{\mathbb{R}_{h-1}} q_{k_h}(\Vec{u}_{k_h}) 
            \xrightarrow{\gamma_h} q_{\ell(k_h)}(\Vec{u}_{\ell(k_h)})
            \xrightarrow{\alpha_h}_{\mathbb{R}_h} q_{n}(\Vec{u}_n)
    \end{split}
    \end{align}
    where 
    \begin{itemize}
        \item $k_1 = \min K$, $k_{i + 1} = \Next(\ell(k_i))$ for all $1\le i < h$, $\ell(k_h) = \max K$;
        \item $\mathbb{R}_0, \ldots, \mathbb{R}_h \in \{ \mathbb{L} , \mathbb{O} \}$.
    \end{itemize}

    By Lemma \ref{lem:flat-high-runs} and \ref{lem:flat-near-axes}, each $\alpha_i$ can be captured by an LPS $\Lambda_i$ such that $|\Lambda_i|\le |G|^{O(1)}$. By Lemma \ref{lem:flat-high-cycles}, each $\gamma_i$ is also captured by an LPS $\Gamma_i$ such that $|\Gamma_i|\le |G|^{O(1)}$. Note that $h \le |Q|$. Let
    \begin{equation}
        \Lambda := \Lambda_0\Gamma_1\Lambda_1 \ldots \Gamma_h\Lambda_h.
    \end{equation}
    Then $\Lambda$ is compatible to $G$ and $p(\Vec{u}) \xrightarrow{\Lambda} q(\Vec{v})$. As for the length of $|\Lambda|$, just note that $|\Lambda| \le (2|Q| + 1) \cdot |G|^{O(1)} \le |G|^{O(1)}$.
\end{proof}

\section{Proof of Lemma \ref{lem:LPS-system}}

\lemLPSSystem*

\begin{proof}
    The ``only if'' part is straightforward: assuming $p(\Vec{u}) \xrightarrow{\Lambda} q(\Vec{v})$, there exist numbers $e_1, \ldots, e_k\in\mathbb{N}_{>0}$ such that $\pi: p(\Vec{u}) \xrightarrow{\alpha_0\beta_1^{e_1}\alpha_1\ldots \beta_k^{e_k}\alpha_k} q(\Vec{v})$. Let $\Vec{e} \in\mathbb{N}^k$ such that $\Vec{e}(i) := e_i$, one easily sees that $(\Vec{u}, \Vec{e}, \Vec{v}) \models \LPSSystem{\Lambda}$ since $\LPSSystem{\Lambda}$ merely asserts that a small subset of the configurations visited by the run $\pi$ stay in $\mathbb{N}^d$, and that $\Delta(\pi) = \Vec{v} - \Vec{u}$. 
    
    Now we focus on the ``if'' part. Assume $(\Vec{u}, \Vec{e}, \Vec{v}) \models \LPSSystem{\Lambda}$ for some $\Vec{e} = (e_1, \ldots, e_k) \in \mathbb{N}_{>0}^k$. We claim that 
    \begin{equation}
        p(\Vec{u}) \xrightarrow{\alpha_0\beta_1^{e_1}\alpha_1\ldots \beta_k^{e_k}\alpha_k}_{\mathbb{N}^d} q(\Vec{v}).
    \end{equation}
    It is clear that $\Vec{u} + \Delta(\alpha_0\beta_1^{e_1}\alpha_1\ldots \beta_k^{e_k}\alpha_k) = \Vec{v}$, which is asserted in the fourth condition of $\LPSSystem{\Lambda}$. To prove that every configuration along this walk stay in $\mathbb{N}^d$, we only need to check that for every $i = 1, \ldots, k$ with $e_i > 1$, every $\ell = 1, \ldots, e_i - 1$ and every $j = 1, \ldots, |\beta_i|$, 
    \begin{equation}
        \label{eq:lem-lps-system-goal}
        \Vec{u} + \Delta(\alpha_0\beta_1^{e_1}\alpha_1\ldots \beta_{i-1}^{e_{i-1}}\alpha_{i-1}\beta_i^{\ell}) + \Delta(\beta_i[1 \ldots j]) \ge \Vec{0}.
    \end{equation}
    Note that by the third condition of $\LPSSystem{\Lambda}$, we have 
    \begin{alignat}{3}
        \Delta_0 :={} &\Vec{u} + \Delta(\alpha_0\beta_1^{e_1}\alpha_1\ldots \beta_{i-1}^{e_{i-1}}\alpha_{i-1}) + 0\cdot\Delta(\beta_i) &{}+{} \Delta(\beta_i[1 \ldots j]) \ge \Vec{0},\\
        \Delta_{1} :={} &\Vec{u} + \Delta(\alpha_0\beta_1^{e_1}\alpha_1\ldots \beta_{i-1}^{e_{i-1}}\alpha_{i-1}) + (e_i-1)\cdot\Delta(\beta_i) &{}+{} \Delta(\beta_i[1 \ldots j])  \ge \Vec{0}.
    \end{alignat}
    Then 
    \begin{align}
    \begin{split}
        & \Vec{u} + \Delta(\alpha_0\beta_1^{e_1}\alpha_1\ldots \beta_{i-1}^{e_{i-1}}\alpha_{i-1}\beta_i^{\ell}) + \Delta(\beta_i[1 \ldots j]) \\
        ={} & \Vec{u} + \Delta(\alpha_0\beta_1^{e_1}\alpha_1\ldots \beta_{i-1}^{e_{i-1}}\alpha_{i-1}) + \ell\cdot\Delta(\beta_i) + \Delta(\beta_i[1 \ldots j])\\
        ={} & \frac{e_i - 1 - \ell}{e_i - 1} \cdot \Delta_0 + \frac{\ell}{e_i - 1} \cdot \Delta_1
        \ge{} \Vec{0}.
    \end{split}
    \end{align}
    This proves (\ref{eq:lem-lps-system-goal}).
\end{proof}
\section{Proofs Omitted from Section \ref{sec:mod-klmst}}

\subsection{Proof of Lemma \ref{lem:unsat-klm-empty-act-lang}}

\unsatKLMEmptyActLang*

\begin{proof}
    Let $\xi=\bracket{p_0(\Vec{x}_0)G_0q_0(\Vec{y}_0)} \Lambda_1 \cdots \Lambda_k \bracket{p_k(\Vec{x}_k) G_k q_k(\Vec{y}_k)}$. If $L_\xi\ne\emptyset$, then there is a run $\pi$ admitted by $\xi$ of the form 
    \begin{equation}
        p_0(\Vec{m}_0) \xrightarrow{\pi_0} q_0(\Vec{n}_0)
        \xrightarrow{\rho_1}
        \cdots \xrightarrow{\rho_k}
        p_k(\Vec{m}_k) \xrightarrow{\pi_k} q_k(\Vec{n}_k).
    \end{equation}
    Let $\phi_i$ be the Parikh image of $\pi_i$ for each $i = 0, \ldots, k$. Since $q_{i-1}(\Vec{n}_{i-1}) \xrightarrow{\rho_i} p_i(\Vec{m}_i)$ is admitted by $\Lambda_i$, by Lemma \ref{lem:LPS-system} there exists vector $\Vec{e}_i \in \mathbb{N}^{|\Lambda_i|_*}$ such that $(\Vec{n}_{i-1}, \Vec{e}_i, \Vec{m}_i) \models \LPSSystem{\Lambda_i}$. Now the sequence $(\Vec{m}_0, \phi_0, \Vec{n}_0), \Vec{e}_1, \ldots, \Vec{e}_k, (\Vec{m}_k, \phi_k, \Vec{n}_k)$ is a model of $\KLMSystem{\xi}$.
\end{proof}

\subsection{Proof of Lemma \ref{lem:klm-eq-bounded-bounds}}

\klmEquationBoundedBounds*

We first recall some facts in linear algebra.

\begin{lemma}[{\cite[Prop.\ 4]{CH16}}, see also \cite{Pottier91}]
    \label{lem:CH16}
    Let $A \in \mathbb{Z}^{m \times n}$ be an integer matrix and $\Vec{b} \in \mathbb{Z}^m$ be an integer vector. Define the following two sets:
    \begin{equation}
        \Vec{X} := \left\{ \Vec{x} \in \mathbb{N}^n : A\Vec{x} = \Vec{b} \right\},
        \qquad 
        \Vec{X}_0 := \left\{ \Vec{x} \in \mathbb{N}^n : A\Vec{x} = \Vec{0} \right\}.
    \end{equation}
    Then every vector in $\Vec{X}$ can be decomposed as the sum of a vector $\Vec{x} \in \Vec{X}$ and a finite sum of vectors $\Vec{x}_0 \in \Vec{X}_0$ such that 
    \begin{equation}
        \norm{\Vec{x}} \le ((n + 1)\norm{A} + \norm{\Vec{b}} + 1)^m, \qquad 
        \norm{\Vec{x}_0} \le (n \norm{A} + 1)^m.
    \end{equation}
\end{lemma}

As linear inequalities can be converted to linear equations by adding new variables, we can obtain a similar result for systems of linear inequalities.

\begin{corollary}
    \label{cor:mod-pottier}
    Let $A \in \mathbb{Z}^{m \times n}$ be an integer matrix and $\Vec{b} \in \mathbb{Z}^m$ be an integer vector. Define the following two sets:
    \begin{equation}
        \Vec{X} := \left\{ \Vec{x} \in \mathbb{N}^n : A\Vec{x} \le \Vec{b} \right\},
        \qquad 
        \Vec{X}_0 := \left\{ \Vec{x} \in \mathbb{N}^n : A\Vec{x} \le \Vec{0} \right\}.
    \end{equation}
    Then every vector in $\Vec{X}$ can be decomposed as the sum of a vector $\Vec{x} \in \Vec{X}$ and a finite sum of vectors $\Vec{x}_0 \in \Vec{X}_0$ such that 
    \begin{equation}
        \norm{\Vec{x}} \le ((n + m + 1)\norm{A} + \norm{\Vec{b}} + 1)^m, \qquad 
        \norm{\Vec{x}_0} \le ((n + m) \norm{A} + 1)^m.
    \end{equation}
\end{corollary}

\begin{proof}
    The corollary holds trivially when $A = 0$, so we assume $A \ne 0$ here, especially $\norm{A} \ge 1$. Let's define the following sets where $I_m \in \mathbb{Z}^{m\times m}$ is the identity matrix:
    \begin{align}
        \Vec{X}^+ := \biggl\{ (\Vec{x}, \Vec{y}) \in \mathbb{N}^n \times \mathbb{N}^m : 
            \begin{bmatrix}A & I_m\end{bmatrix} 
            \begin{bmatrix} \Vec{x} \\ \Vec{y} \end{bmatrix} 
            = \Vec{b} 
        \biggr\},\\
        \Vec{X}_0^+ := \biggl\{ (\Vec{x}, \Vec{y}) \in \mathbb{N}^n \times \mathbb{N}^m : 
            \begin{bmatrix}A & I_m\end{bmatrix} 
            \begin{bmatrix} \Vec{x} \\ \Vec{y} \end{bmatrix} 
            = \Vec{0} 
        \biggr\}.
    \end{align}
    Let $\Vec{x} \in \Vec{X}$, there exists $\Vec{y} \in \mathbb{N}^m$ with $(\Vec{x}, \Vec{y}) \in \Vec{X}^+$. By Lemma \ref{lem:CH16}, there exists vectors $(\Vec{x}', \Vec{y}') \in \Vec{X}^+$ and $(\Vec{x}_1, \Vec{y}_1), \ldots, (\Vec{x}_k, \Vec{y}_k) \in \Vec{X}_0^+$ such that 
    \begin{equation}
        \Vec{x} = \Vec{x}' + \sum_{i = 1}^k \Vec{x}_i, \quad 
        \Vec{y} = \Vec{y}' + \sum_{i = 1}^k \Vec{y}_i,
    \end{equation}
    where $\norm{\Vec{x}'} \le ((n + m + 1)\norm{A} + \norm{\Vec{b}} + 1)^m$, and $ 
    \norm{\Vec{x}_i} \le ((n + m)\norm{A} + 1)^m $ for all $i = 1, \ldots, k$.
    Clearly $\Vec{x}' \in \Vec{X}$ and $\Vec{x}_i \in \Vec{X}_0$ for all $i = 1, \ldots, k$. This proves the desired result.
\end{proof}

The following lemma is just an application of Corollary \ref{cor:mod-pottier} to the characteristic systems of linear KLM sequences:

\begin{lemma}
    \label{lem:decomp-model-eq-klm}
    Let $\xi$ be a linear KLM sequence. Then any model of $\KLMSystem{\xi}$ can be decomposed into the sum of a model $\Vec{h}$ of $\KLMSystem{\xi}$ and a finite sum of models $\Vec{h}_0$ of $\KLMSystemHomo{\xi}$, such that
    \begin{equation}
        \norm{\Vec{h}}, \norm{\Vec{h}_0} \le (10|\xi|)^{12|\xi|}.
    \end{equation}
\end{lemma}

\begin{proof}
    Suppose $\xi = \bracket{p_0(\Vec{x}_0)G_0q_0(\Vec{y}_0)} \Lambda_1 \cdots \Lambda_k \bracket{p_k(\Vec{x}_k) G_k q_k(\Vec{y}_k)}$ where $G_i = (Q_i, T_i)$ for each $i = 0, \ldots, k$. In order to apply Corollary \ref{cor:mod-pottier} to $\KLMSystem{\xi}$, we need to bound the following four parameters:
    \begin{itemize}
        \item The number of variables in $\KLMSystem{\xi}$ is 
            \begin{equation}
                2d\cdot(k + 1) + \sum_{i = 0}^{k}|T_i| + \sum_{i = 1}^{k} |\Lambda_i|_* \le 2|\xi|.
            \end{equation}
        \item The number of inequalities in $\KLMSystem{\xi}$ is bounded by 
            \begin{equation}
                6d\cdot (k + 1) + 2\sum_{i = 0}^{k}|T_i| + \sum_{i = 1}^{k}(|\Lambda_i|_* + d\cdot (2|\Lambda_i| + 2)) \le 6|\xi|.
            \end{equation}
            To see this, note that each Kirchoff system $K_{G_i, p_i, q_i}$ consists of $|T_i|$ equations; each $\LPSSystem{\Lambda_i}$ consists of at most $|\Lambda_i|_* + 2d \cdot |\Lambda_i|$ inequalities and one equation; and the constraints $\Vec{m}_i \sqsubseteq \Vec{x}_i$, $\Vec{n}_i \sqsubseteq \Vec{y}_i$, and $\Vec{n}_i = \Vec{m}_i + \Delta(\phi_i)$ each consists of $d$ equations. Also remember that one equation corresponds to two inequalities.
        \item The maximum absolute value of the coefficients in $\KLMSystem{\xi}$ is bounded by 
            \begin{equation}
                \max\left\{ 1, \max_{0\le i\le k} \norm{T_i}, \max_{1\le i\le k} \norm{\Lambda_i} \right\} \le |\xi|.
            \end{equation}
            Note that the coefficients in $K_{G_i, p_i, q_i}$ and in constraints $\Vec{m}_i \sqsubseteq \Vec{x}_i$, $\Vec{n}_i \sqsubseteq \Vec{y}_i$ are just $\pm 1$. The coefficients in $\LPSSystem{\Lambda_i}$ are bounded by $\norm{\Lambda_i}$, and that in the constraints $\Vec{n}_i = \Vec{m}_i + \Delta(\phi_i)$ are bounded by $\norm{T_i}$.
        \item The maximum absolute value of the constant terms in $\KLMSystem{\xi}$ is bounded by 
            \begin{equation}
                \max\left\{ 1, \max_{0\le i\le k} \norm{\Vec{x}_i}, \max_{0\le i\le k} \norm{\Vec{y}_i}, \max_{1\le i\le k} \norm{\Lambda_i} \right\} \le |\xi|.
            \end{equation}
            Note that the constraints $\Vec{n}_i = \Vec{m}_i + \Delta(\phi_i)$ contain no constant terms. The constant terms in constraints $\Vec{m}_i \sqsubseteq \Vec{x}_i$, $\Vec{n}_i \sqsubseteq \Vec{y}_i$ are the non-$\omega$ components of $\Vec{x}_i$ and $\Vec{y}_i$, and that in $K_{G_i, p_i, q_i}$ are just $\pm 1$. Finally, the constant terms in $\LPSSystem{\Lambda_i}$ are bounded by $\norm{\Lambda_i}$.
    \end{itemize}

    Now by Corollary \ref{cor:mod-pottier} we can bound $\norm{\Vec{h}}$ and $\norm{\Vec{h}_0}$ by 
    \begin{equation*}
        \norm{\Vec{h}}, \norm{\Vec{h}_0} \le ((2|\xi| + 6|\xi| + 1)\cdot |\xi| + |\xi| + 1)^{6|\xi|} \le (10 |\xi|)^{12|\xi|}.
        \qedhere
    \end{equation*}
\end{proof}

Lemma \ref{lem:klm-eq-bounded-bounds} now follows immediately from Lemma \ref{lem:decomp-model-eq-klm}.

\subsection{Proof of Lemma \ref{lem:decom-scc}}

\decomSCC*

\begin{proof}
    For each tuple $\zeta = \bracket{p(\Vec{x})Gq(\Vec{y})}$ occurring in $\xi$ where $G$ is not strongly connected, we replace it by all possible sequences of the form 
    \begin{equation}
        \zeta' = \bracket{p(\Vec{x}) G_0 r_0(\Vec{\omega})} t_1 \bracket{s_1(\Vec{\omega}) G_1 r_1(\Vec{\omega})} t_2\ldots \bracket{s_{n-1}(\Vec{\omega}) G_{n-1} r_{n-1}(\Vec{\omega})} t_n\bracket{s_n(\Vec{\omega}) G_n q(\Vec{y})}
    \end{equation}
    where $G_0, \ldots, G_n$ are distinct strongly connected components of $G$ connected by transitions $t_1, \ldots, t_n$. Clearly the action language is preserved, and the rank does not increase. As for the size, note that $n \le |\zeta|$ and thus
    \begin{equation}
        |\zeta'| = \sum_{i = 0}^{n} |G_i| + d\cdot(\norm{\Vec{x}} + \norm{\Vec{y}} + 1) + d\cdot n + \sum_{i = 1}^{n} d\cdot (\norm{t_i} + 1) \le |\zeta| + 2nd \le (2d + 1)|\zeta|.
    \end{equation}
    This bounds the amplification in sizes.
\end{proof}

\subsection{Proof of Lemma \ref{lem:decom-pure}}

\decomPure*

\begin{proof}
    Since $\xi$ is strongly connected, by Corollary \ref{cor:sc-vass-rank0-iff-geo-2d}, to decide whether $\xi$ is pure it is enough to check whether there is a non-trivial tuple $\bracket{p(\Vec{x}) G q(\Vec{y})}$ in $\xi$ such that $G$ is geometrically $2$-dimensional, i.e. $\dim(\CycleSpace(G)) \le 2$. Observe that $\CycleSpace(G)$ is indeed spanned by the effects of all simple cycles in $G$, thus its dimension can be computed by enumerating the simple cycles, which can be done in $\mathsf{PSPACE}$.

    Suppose $\xi$ is not pure, then for each non-trivial tuple $\bracket{p(\Vec{x}) G q(\Vec{y})}$ occurring in $\xi$ such that $\Rank(G) = \Vec{0}$, we replace it by all possible sequences of the form
    \begin{equation}
        \label{eq:impure-to-trivial-lps}
        \bracket{p(\Vec{x})} \Lambda \bracket{q(\Vec{y})}
    \end{equation}
    where $\Lambda$ is a positive LPS compatible to $G$ from state $p$ to $q$ satisfying $|\Lambda| \le |G|^{O(1)}$ where the $O(1)$-term comes from Theorem \ref{thm:geo-2vass-flat-plps}. Let the resulting set of linear KLM sequences be $\Xi$, we claim that $\Xi$ satisfies the lemma. Clearly every sequence $\xi'$ in $\Xi$ is pure. Note that $\Rank(\xi') = \Rank(\xi)$ since we only replace rank-$\Vec{0}$ tuples by rank-$\Vec{0}$ sequences. Also note that as $|\Lambda| \le |G|^{O(1)}$ and $\norm{\Lambda} \le |G|$ (because $\Lambda$ is compatible to $G$), we have $|\xi'| \le |\xi|^{O(1)}$ for all $\xi'\in \Xi$. Now we focus on the action languages.

    Since we replace rank-$\Vec{0}$ VASSes by positive LPSes compatible to them, every run admitted by the LPS must be admitted by the original VASS. Thus the inclusion $\bigcup_{\xi' \in \Xi}L_{\xi'} \subseteq L_\xi$ is clear.
    Let $\xi$ be given by $\bracket{p_0(\Vec{x}_0)G_0q_0(\Vec{y}_0)} \Lambda_1 \cdots \Lambda_k \bracket{p_k(\Vec{x}_k) G_k q_k(\Vec{y}_k)}$. If $L_\xi \ne \emptyset$, consider a run $\pi = \pi_0\rho_1\pi_1\ldots \rho_k\pi_k$ admitted by $\xi$ of the form
    \begin{equation}
        p_0(\Vec{m}_0) \xrightarrow{\pi_0} q_0(\Vec{n}_0) \xrightarrow{\rho_1} p_1(\Vec{m}_1) \xrightarrow{\pi_1} q_1(\Vec{n}_1) \xrightarrow{\rho_2} \cdots \xrightarrow{\rho_k} p_k(\Vec{m}_k) \xrightarrow{\pi_k} q_k(\Vec{n}_k)
    \end{equation}
    where $\Vec{m}_i \sqsubseteq \Vec{x}_i$ and $\Vec{n}_i \sqsubseteq \Vec{y}_i$. We will exhibit a run $\overline{\pi}$ admitted by some $\xi' \in \Xi$. For each $i = 0, \ldots, k$, if $\Rank(G_i) \ne \Vec{0}$ or if the KLM tuple $\xi_i$ is trivial, we let $\overline{\pi_i} = \pi_i$. Otherwise, $\Rank(G_i) = \Vec{0}$ and $G_i$ is geometrically $2$-dimensional by Corollary \ref{cor:sc-vass-rank0-iff-geo-2d}. Since $p_i(\Vec{m}_i)\xrightarrow{\pi_i} q_i(\Vec{n}_i)$ holds in $G_i$, by Theorem \ref{thm:geo-2vass-flat-plps} there is a positive LPS $\Lambda_i$ compatible to $G_i$ with $|\Lambda_i| \le |G|^{O(1)}$ such that $p_i(\Vec{m}_i)\xrightarrow{\Lambda_i} q_i(\Vec{n}_i)$. In this case we let $\overline{\pi_i}$ be a path admitted by $\Lambda_i$ so that $p_i(\Vec{m}_i)\xrightarrow{\overline{\pi_i}} q_i(\Vec{n}_i)$. Now observe that $\overline{\pi} := \overline{\pi_0}\rho_1\overline{\pi_1}\ldots \rho_k\overline{\pi_k}$ is admitted by some $\xi' \in \Xi$, witnessing $\bigcup_{\xi'\in\Xi}L_{\xi'}\ne\emptyset$.
\end{proof}

\subsection{Proof of Lemma \ref{lem:cleaning}}

\lemCleaning*

\begin{proof}
    We first apply Lemma \ref{lem:decom-scc} to $\xi$ to obtain a set of strongly connected linear KLM sequences with the action language preserved. Then we apply Lemma \ref{lem:decom-pure} to make them pure. The action language is partially preserved. These linear KLM sequences are then saturated using Lemma \ref{lem:decom-satur}, which preserves action language. Finally, as unsatisfiable sequences have empty action languages by Lemma \ref{lem:unsat-klm-empty-act-lang}, we can safely remove them. The ranks of the linear KLM sequences never increased. As for the sizes, note that Lemma \ref{lem:decom-scc} increases $|\xi|$ linearly, while Lemma \ref{lem:decom-pure} increases $|\xi|$ by a polynomial function. After that, Lemma \ref{lem:decom-satur} produces sequences of size bounded by $\PolyF{|\xi|}^{\PolyF{|\xi|}} \le |\xi|^{\PolyF{|\xi|}}$.
\end{proof}

\subsection{Proof of Lemma \ref{lem:pure-rank-full-decrease-imply-rank-decrease}}

\pureRankDecrease*

\begin{proof}
    Let $\RankFull(\xi) = (r_d, \ldots, r_0)$. By the definition of pure sequences and Lemma \ref{lem:sc-vass-cycspace}, we must have $r_2 = r_1 = r_0 = 0$. Now let $\xi'$ be a linear KLM sequence with $\RankFull(\xi') \ltlex \RankFull(\xi)$. Suppose $\RankFull(\xi') = (r_d', \ldots, r_0')$, then there exists $3\le i\le d$ with $r_i' < r_i$ and $r_j' = r_j$ for all $i < j \le d$. This implies $\Rank(\xi')\ltlex \Rank(\xi)$.
\end{proof}

\subsection{Proof of Lemma \ref{lem:klm-seq-decomposition}}

\klmSeqDecomposition*

\begin{proof}
    If a clean linear KLM sequence $\xi$ is not normal, it is either bounded, unpumpable, or not rigid. We apply one of Lemmas \ref{lem:decom-unb}, \ref{lem:decom-rigid}, and \ref{lem:decom-pump} to decompose $\xi$ into a set $\Xi$ of linear KLM sequences such that $\bigcup_{\xi'\in\Xi} L_{\xi'} = L_\xi$, and $\Rank(\xi') \ltlex \Rank(\xi)$ and $|\xi'| \le |\xi|^{O(|\xi|)}$ for every $\xi'\in\Xi$. Then apply Lemma \ref{lem:cleaning} to clean each sequence $\xi' \in \Xi$. We are done by setting $\Dec{\xi} := \bigcup_{\xi' \in \Xi}\Clean{\xi'}$.
\end{proof}

\subsection{Proof of Lemma \ref{lem:klm-normal-bounded-witness}}

\klmNormalBoundedWitness*

The proof of \ref{lem:klm-normal-bounded-witness} is basically just a repetition of that of \cite[Lem.\ 4.19]{LS19}, except that we need to take additional care of the part of linear path schemes. Fix a normal linear KLM sequence $\xi = \bracket{p_0(\Vec{x}_0)G_0q_0(\Vec{y}_0)} \Lambda_1 \cdots \Lambda_k \bracket{p_k(\Vec{x}_k) G_k q_k(\Vec{y}_k)}$ throughout this section, where $G_i = (Q_i, T_i)$ for each $i = 0, \ldots, k$. First we need a bound on the models of $\KLMSystem{\xi}$ and $\KLMSystemHomo{\xi}$.

\begin{claim}
    \label{clm:normal-bounds-klm-homo-sys}
    There is a model $\Vec{h}_0 \models \KLMSystemHomo{\xi}$ with $\norm{\Vec{h}_0} \le |\xi|^{O(|\xi|)}$ such that for every $i = 0, \ldots, k$, every $j \in [d]$ and every $t\in T_i$, 
    \begin{itemize}
        \item $\Vec{m}_i^{\Vec{h}_0}(j) > 0$ whenever $\Vec{x}_i(j) = \omega$,
        \item $\Vec{n}_i^{\Vec{h}_0}(j) > 0$ whenever $\Vec{y}_i(j) = \omega$,
        \item $\phi_i^{\Vec{h}_0}(t) > 0$.
    \end{itemize}
\end{claim}

\begin{claimproof}
    Since $\xi$ is saturated and unbounded, we are done by taking $\Vec{h}_0$ as a sum of at most $2d(k + 1) + \sum_{i = 0}^{k}|T_i| \le 2|\xi|$ models of $\KLMSystemHomo{\xi}$ each of whose norm is bounded by $|\xi|^{O(|\xi|)}$ by Lemma \ref{lem:klm-eq-bounded-bounds}.
\end{claimproof}

\begin{claim}
    \label{clm:normal-bounds-klm-sys}
    There is a model $\Vec{h} \models \KLMSystem{\xi}$ with $\norm{\Vec{h}} \le |\xi|^{O(|\xi|)}$ such that $\phi_i^{\Vec{h}}(t) > 0$ for every $i = 0, \ldots, k$ and every $t\in T_i$.
\end{claim}

\begin{claimproof}
    By Lemma \ref{lem:decomp-model-eq-klm} there is a model $\Vec{h}' \models \KLMSystem{\xi}$ with $\norm{\Vec{h}'} \le |\xi|^{O(|\xi|)}$. We are done by taking $\Vec{h} = \Vec{h}' + \Vec{h}_0$ where $\Vec{h}_0$ is given by Claim \ref{clm:normal-bounds-klm-homo-sys}.
\end{claimproof}

We also need to borrow the following result from \cite{LS19}, which bounds the lengths of the runs that ``pump'' each input and output configurations in $\xi$.

\begin{claim}[{\cite[Claim C.4]{LS19}}]
    \label{clm:normal-bounds-pump-runs}
    For each $i = 0, \ldots, k$ there are paths $u_i, v_i$ of $G_i$ with $|u_i|, |v_i| \le |\xi|^{O(1)}$ and vectors $\Vec{x}_i', \Vec{y}_i' \in \mathbb{N}_\omega^d$ with $\Vec{x}_i' \ge \Vec{x}_i, \Vec{y}_i' \ge \Vec{y}_i$ such that 
    \begin{itemize}
        \item $p_i(\Vec{x}_i) \xrightarrow{u_i} p_i(\Vec{x}_i')$ and $\Vec{x}_i'(j) > \Vec{x}_i(j)$ whenever $\Facc_{G_i, p_i}(\Vec{x}_i)(j) = \omega \ne \Vec{x}_i(j)$ for every $j\in[d]$,
        \item $q_i(\Vec{y}_i') \xrightarrow{v_i} q_i(\Vec{y}_i)$ and $\Vec{y}_i'(j) > \Vec{y}_i(j)$ whenever $\Bacc_{G_i, q_i}(\Vec{y}_i)(j) = \omega \ne \Vec{y}_i(j)$ for every $j\in[d]$.
    \end{itemize}
\end{claim}

Now we can prove Lemma \ref{lem:klm-normal-bounded-witness}.

\begin{proof}[Proof of Lemma \ref{lem:klm-normal-bounded-witness}]
    Let $\Vec{h}$ and $\Vec{h}_0$ be the models of $\KLMSystem{\xi}$ and $\KLMSystemHomo{\xi}$ given by Claim \ref{clm:normal-bounds-klm-sys} and \ref{clm:normal-bounds-klm-homo-sys}. For each $i = 0, \ldots, k$ let the paths $u_i, v_i$ and the vectors $\Vec{x}_i', \Vec{y}_i'$ be given by Claim \ref{clm:normal-bounds-pump-runs}. Denote by $\psi_{u_i}$ and $\psi_{v_i}$ the Parikh image of $u_i$ and $v_i$, respectively.

    Let $\phi_i := r_i\phi_i^{\Vec{h}_0} - (\psi_{u_i} + \psi_{v_i})$ where we choose $r_i := \norm{T_i} \cdot (|u_i| + |v_i|) + 1$. By Claim \ref{clm:normal-bounds-pump-runs} we have $r_i \le |\xi|^{O(|\xi|)}$. Note that $\phi_i(t) > 0$ for every $t\in T_i$. Now observe that $\phi_i \models K_{G_i, p_i, q_i}^0$. Since $G_i$ is strongly connected, by Euler's Lemma there is a path $w_i$ in $G_i$ whose Parkih image is $\phi_i$. We have $|w_i| = \sum_{t\in T_i}\phi_i(t) \le r_i|\xi|^{O(|\xi|)} \le |\xi|^{O(|\xi|)}$. Also note that 
    \begin{equation}
        \Delta(w_i) = \Delta(\phi_i) = r_i\Delta(\phi_i^{\Vec{h}}) - (\Delta(u_i) + \Delta(v_i)) = r_i(\Vec{n}_i^{\Vec{h}_0} - \Vec{m}_i^{\Vec{h}_0}) - (\Delta(u_i) + \Delta(v_i)).
    \end{equation}
    Let $\Vec{m}_i^0 := r_i\Vec{m}_i^{\Vec{h}_0} + \Delta(u_i)$ and $\Vec{n}_i^0 := r_i\Vec{n}_i^{\Vec{h}_0} - \Delta(v_i)$, then $\Delta(w_i) = \Vec{n}_i^0 - \Vec{m}_i^0$. Now we prove that $\Vec{m}_i^0, \Vec{n}_i^0 \ge \Vec{0}$ and $\Vec{m}_i^0(j), \Vec{n}_i^0(j) > 0$ for every coordinate $j$ not fixed by $G_i$. For coordinates $j$ fixed by $G_i$ we know that $\Delta(u_i)(j) = \Delta(v_i)(j) = 0$ since $u_i, v_i$ are cycles, so $\Vec{m}_i^0(j) = r_i\Vec{m}_i^{\Vec{h}_0}(j) \ge 0$ and $\Vec{n}_i^0(j) = r_i\Vec{n}_i^{\Vec{h}_0}(j) \ge 0$. For coordinates $j$ not fixed by $G_i$, if $\Vec{x}_i(j) \in \mathbb{N}$ then $\Facc_{G_i, p_i}(\Vec{x}_i)(j) = \omega$ and $\Delta(u_i)(j) > 0$, so $\Vec{m}_i^0(j) \ge \Delta(u_i)(j) > 0$; if $\Vec{x}_i(j) = \omega$ then $\Vec{m}_i^{\Vec{h}_0}(j) > 0$ and we have $\Vec{m}_i^0(j) \ge r - \Delta(u_i)(j) > 0$ since $r > \norm{T_i}\cdot |u_i| \ge \Delta(u_i)(j)$. Symmetrically we can prove $\Vec{n}_i^0(j) > 0$ for $j$ not fixed by $G_i$.

    Since $u_i$ is a run from $\Vec{x}_i = \Vec{m}_i^{\Vec{h}} + \omega \Vec{m}_i^{\Vec{h}_0}$, and $\Delta(u)_i(j) \le \norm{T_i}\cdot |u_i| < r_i$ for every $j\in [d]$, we deduce that $u_i$ can be fired from the configuration $p_i(\Vec{m}_i^{\Vec{h}} + r_i \Vec{m}_i^{\Vec{h}_0})$. Then we have 
    \begin{equation}
        p_i(\Vec{m}_i^{\Vec{h}} + r_i \Vec{m}_i^{\Vec{h}_0}) \xrightarrow{u_i}
        p_i(\Vec{m}_i^{\Vec{h}} + \Vec{m}_i^0).
    \end{equation}
    Let $s_i := 1 + \norm{T_i}\cdot \max\{|w_i|, |T_i|\cdot\norm{\Vec{h}}\} \le |\xi|^{O(|\xi|)}$, by monotony we have
    \begin{equation}
        \label{eq:klm-pump-up}
        p_i(\Vec{m}_i^{\Vec{h}} + s_ir_i \Vec{m}_i^{\Vec{h}_0}) \xrightarrow{u_i^{s_i}}
        p_i(\Vec{m}_i^{\Vec{h}} + s_i\Vec{m}_i^0).
    \end{equation}
    Symmetrically, we have
    \begin{equation}
        \label{eq:klm-pump-down}
        q_i(\Vec{n}_i^{\Vec{h}} + s_i\Vec{n}_i^0) \xrightarrow{u_i^{s_i}}
        q_i(\Vec{n}_i^{\Vec{h}} + s_ir_i \Vec{n}_i^{\Vec{h}_0}).
    \end{equation}
    Note that $s_i > |w_i|\cdot \norm{T_i} \ge \Delta(w_i)(j)$ for all $j\in[d]$ and that $\Vec{m}_i^0(j), \Vec{n}_i^0(j) > 0$ for every coordinate $j$ not fixed by $G_i$, we deduce 
    \begin{equation}
        \label{eq:klm-compensate}
        p_i(\Vec{m}_i^{\Vec{h}} + s_i\Vec{m}_i^0) \xrightarrow{w_i^{s_i}} 
        p_i(\Vec{m}_i^{\Vec{h}} + s_i\Vec{n}_i^0).
    \end{equation}

    Since $\phi_i^{\Vec{h}} \models K_{G_i, p_i, q_i}$ and $\phi_i^{\Vec{h}}(t) > 0$ for every $t\in T_i$, again by Euler's Lemma we know there exists a path $\sigma_i$ from $p$ to $q$ whose Parikh image is $\phi_i^{\Vec{h}}$. Observe that 
    \begin{equation}
        p_i(\Facc_{G_i, p_i}(\Vec{x}_i)) \xrightarrow{\sigma_i} q_i(\Bacc_{G_i, q_i}(\Vec{y}_i)).
    \end{equation}
    Also note that $\Facc_{G_i, p_i}(\Vec{x}_i) = \Vec{m}_i^{\Vec{h}} + \omega\Vec{n}_i^0$ and $\Bacc_{G_i, q_i}(\Vec{y}_i) = \Vec{n}_i^{\Vec{h}} + \omega\Vec{n}_i^0$. Since $|\sigma_i| = \sum_{t\in T_i}\phi_i^{\Vec{h}}(t) \le |T_i|\norm{\Vec{h}}$, we have $s_i > |\sigma|\cdot\norm{T_i} \ge \Delta(\sigma_i)(j)$ for all $j\in[d]$. Thus 
    \begin{equation}
        \label{eq:klm-work}
        p_i(\Vec{m}_i^{\Vec{h}} + s_i\Vec{n}_i^0) \xrightarrow{\sigma_i} q_i(\Vec{n}_i^{\Vec{h}} + s_i\Vec{n}_i^0).
    \end{equation}

    Let $s := \max_{0\le i\le k}s_i$ and $r := \max_{0\le i\le k}r_i$, one can verify that the above discussion is still valid with $r_i, s_i$ replaced by $r, s$.
    Combining (\ref{eq:klm-pump-up}), (\ref{eq:klm-compensate}), (\ref{eq:klm-work}) and (\ref{eq:klm-pump-down}), let $\pi_i := u_i^{s}w_i^{s}\sigma_i v_i^{s}$, we have 
    \begin{equation}
        \label{eq:normal-klm-tuple-admit-run}
        p_i(\Vec{m}_i^{\Vec{h}} + sr \Vec{m}_i^{\Vec{h}_0}) \xrightarrow{\pi_i}
        q_i(\Vec{n}_i^{\Vec{h}} + sr \Vec{n}_i^{\Vec{h}_0})
    \end{equation}
    holds for all $i = 0, \ldots, k$. Observe that $|\pi_i| \le |\xi|^{O(|\xi|)}$.

    Let $\Vec{h}' := \Vec{h} + sr \Vec{h}_0$, then $\Vec{h'}\models\KLMSystem{\xi}$, especially $(\Vec{n}_{i-1}^{\Vec{h}'}, \Vec{e}_i^{\Vec{h}'}, \Vec{m}_i^{\Vec{h}'}) \models \LPSSystem{\Lambda_i}$ for each $i = 1, \ldots, k$. By Lemma \ref{lem:LPS-system} there exists a path $\rho_i$ admitted by $\Lambda_i$ such that 
    \begin{equation}
        \label{eq:normal-klm-lps-admit-run}
        q_{i-1}(\Vec{n}_{i-1}^{\Vec{h}} + sr \Vec{n}_{i-1}^{\Vec{h}_0}) \xrightarrow{\rho_i}
        p_i(\Vec{m}_i^{\Vec{h}} + sr \Vec{m}_i^{\Vec{h}_0}).
    \end{equation}
    Note that $|\rho_i| \le |\Lambda_i| + |\Lambda_i|_*\cdot |\Lambda_i| \cdot \norm{\Vec{e}_i^{\Vec{h}'}} \le |\xi|^{O(|\xi|)}$.

    Combining (\ref{eq:normal-klm-tuple-admit-run}) and (\ref{eq:normal-klm-lps-admit-run}), we know that the run $\pi := \pi_0\rho_1\pi_1\ldots \rho_k\pi_k$ is admitted by $\xi$, thus $\ActWord{\pi} \in L_\xi \ne\emptyset$. Observe that $|\pi| \le (2k + 1)\cdot |\xi|^{O(|\xi|)} \le |\xi|^{O(|\xi|)}$.
\end{proof}

\section{Proof of Lemma \ref{lem:rel-fast-grow-func}}

\relFastGrowFunc*

\begin{proof}
    We need to use the fact that if $h$ is a monotone inflationary function, then so are the functions $h^\alpha$ in the Hardy hierarchy, which can be proved by induction on $\alpha$. Let $G := H^{\omega^b\cdot c}$, first we show the following inequality holds for all $x \ge 2c$:
    \begin{equation}
        \label{eq:fast-grow-func-claim}
        H^{\omega^{b + a}}((c + 1) x) \ge (c + 1)G^{\omega^a}(x).
    \end{equation}
    \begin{claimproof}
        We proceed by induction on $a$. For $a = 1$:
        \begin{align*}
            (c + 1) G^{\omega}(x) &= (c + 1) G^{x + 1}(x)\\
            &= (c + 1)H^{\omega^b \cdot c(x + 1)}(x)\\
            &\le H^{\omega^b\cdot c}(H^{\omega^b \cdot c(x + 1)}(x)) 
                & \text{since }H^{\omega^b\cdot c}(x) \ge H^{\omega\cdot c}(x) \ge 2^c\cdot x \\
            &= H^{\omega^b \cdot c(x + 2)}(x)\\
            &\le H^{\omega^b \cdot ((c + 1)x + 1)}(x) 
                & \text{by }x \ge 2c\\
            &\le H^{\omega^b \cdot ((c + 1)x + 1)}((c + 1) x) 
                &\text{by monotony}\\
            &= H^{\omega^{b + 1}}((c + 1)x).
        \end{align*}
        Next let's consider the case $a + 1$, assuming (\ref{eq:fast-grow-func-claim}) holds for $a$. We first prove the following inequality holds for all $y \ge 2c$ by induction on $j$:
        \begin{equation}
            \label{eq:fast-grow-func-subclaim}
            H^{\omega^{b + a}\cdot j}((c + 1)y) \ge (c + 1)G^{\omega^a \cdot j}(y).
        \end{equation}
        The base case $j = 0$ holds trivially. For the induction step, note that 
        \begin{align*}
            H^{\omega^{b + a}\cdot (j + 1)}((c + 1)y) 
                &= H^{\omega^{b + a}}(H^{\omega^{b + a}\cdot j}((c + 1)y))\\
                &\ge H^{\omega^{b + a}}((c + 1)G^{\omega^a \cdot j}(y)) 
                    &\text{by induction hypothesis of (\ref{eq:fast-grow-func-subclaim})}\\
                &\ge (c + 1)G^{\omega^a}(G^{\omega^a \cdot j}(y))
                    &\text{by induction hypothesis of (\ref{eq:fast-grow-func-claim})}\\
                &= (c + 1)G^{\omega^a \cdot (j + 1)}(y).
        \end{align*}
        Now we can prove (\ref{eq:fast-grow-func-claim}) for $a + 1$:
        \begin{align*}
            H^{\omega^{b + a + 1}}((c + 1) x) 
                &= H^{\omega^{b + a}\cdot ((c + 1)x + 1)}((c + 1)x)\\
                &\ge H^{\omega^{b + a}\cdot (x + 1)}((c + 1)x)\\
                &\ge (c + 1)G^{\omega^{a}\cdot (x + 1)}(x) &\text{by (\ref{eq:fast-grow-func-subclaim})}\\
                &= (c + 1)G^{\omega^{a + 1}}(x).
        \end{align*}
        This completes the proof of (\ref{eq:fast-grow-func-claim}).
    \end{claimproof}

    Now we return to the proof of the lemma. For all $x \ge \max\{2c, x_0\}$ we have 
    \begin{align*}
        h^{\omega^a}(x) &\le G^{\omega^a}(x) &\text{since $h(x) \le G(x)$ for $x\ge x_0$}\\
        &\le (c + 1)G^{\omega^a}(x)\\
        &\le H^{\omega^{b + a}}((c + 1) x) &\text{by (\ref{eq:fast-grow-func-claim})}.
    \end{align*}
    This completes the proof.
\end{proof}

\end{document}